\documentclass{article}
\usepackage[left=2.5cm, top=3cm, bottom=3cm, right=2.5cm]{geometry}
\usepackage{amsfonts}
\usepackage{amsthm}
\usepackage{amsmath}
\usepackage{amssymb}
\usepackage{enumerate}
\usepackage[numbers,sort&compress]{natbib}
\usepackage[dvips]{graphicx}

\everymath{\displaystyle}

\newcommand{\dd}{\mathrm{d}}

\newcommand{\ud}{\,\dd}
\newcommand{\nbr}[1]{$#1$\nobreakdash-\hspace{0pt}}
\newcommand{\braket}[1]{\langle{#1}\rangle}

\newcommand{\lshad}{[\![}
\newcommand{\rshad}{]\!]}
\newcommand{\sdot}{\,\cdot\,}

\providecommand{\abs}[1]{\lvert#1\rvert}
\providecommand{\Abs}[1]{\left\lvert#1\right\rvert}
\providecommand{\norm}[1]{\lVert#1\rVert}

\DeclareMathOperator{\tr}{tr}
\DeclareMathOperator{\Tr}{Tr}

\DeclareMathOperator{\sinc}{sinc}

\newtheorem{theorem}{Theorem}

\theoremstyle{definition}

\title{Deformation quantization with minimal length}
\author{Ziemowit Doma\'nski\thanks{The author is supported by the Ministry of
Science and Higher Education of Poland, grant number 04/43/DSPB/0094.}\\
\small Institute of Mathematics, Pozna{\'n} University of Technology\\
\small Piotrowo 3A, 60-965 Pozna{\'n}, Poland\\
\small \tt ziemowit.domanski@put.poznan.pl \and Maciej B{\l}aszak\\
\small Faculty of Physics, Division of Mathematical Physics, A. Mickiewicz
University\\
\small Umultowska 85, 61-614 Pozna\'n, Poland\\
\small \tt blaszakm@amu.edu.pl}

\begin{document}
\maketitle

\begin{abstract}
We develop a complete theory of non-formal deformation quantization exhibiting
a nonzero minimal uncertainty in position. An appropriate integral formula for
the star-product is introduced together with a suitable space of functions on
which the star-product is well defined. Basic properties of the star-product
are proved and the extension of the star-product to a certain Hilbert space and
an algebra of distributions is given. A \nbr{C^*}algebra of observables and
a space of states are constructed. Moreover, an operator representation in
momentum space is presented. Finally, examples of position eigenvectors and
states of maximal localization are given.
\\[\baselineskip]
\textbf{Keywords and phrases}: quantum mechanics, deformation quantization,
star-product, minimal length, generalized uncertainty principle
\end{abstract}

\section{Introduction}
\label{sec:1}
The idea that there should exists a minimal length scale was born in 1930s with
the advent of quantum field theory. It was believed that it could help to remove
divergences occurring in quantum theory. The interest in a minimal length
scale was later on renewed when it was suggested that the gravity will play a
significant role at short distances and effectively making it impossible to
measure distances to a precision better than Planck's length
\cite{Snyder:1947,Mead:1964,Mead:1966,Adler:1999}. Recently a minimal length
arose also in many theories of quantum gravity and in string theory
\cite{Amati:1987,Amati:1988,Gross:1988,Amati:1990,Konishi:1990,Maggiore:1993,%
Maggiore:1994,Scardigli:1999,Scardigli:2003}. For a broad review of this topic
and discussion of different implementations of a minimal length scale in quantum
mechanics and quantum field theory we refer to \cite{Hossenfelder:2013}.

A way of achieving minimal length scale is through modification of uncertainty
relations for position and momentum. Such modified uncertainty relations are
referred in the literature as a Generalized Uncertainty Principle and can be
reproduced from modified commutation relations for position and momentum
operators. In \cite{Kempf:1995} authors proposed the following modified
uncertainty relation for position and momentum in a one-dimensional case and
nonrelativistic regime
\begin{equation}
\Delta q \Delta p \geq \frac{\hbar}{2}\left(1 + \beta (\Delta p)^2
    + \beta \braket{\hat{p}}^2\right),
\label{eq:1}
\end{equation}
where $\beta$ is some positive constant. This is the simplest form of the
uncertainty relation which leads to a nonzero minimal uncertainty $\Delta q_0$
in position. More general uncertainty relations were considered in
\cite{Kempf:1992,Kempf:1994,Quesne.Penson:2003,Quesne:2003} exhibiting also a
minimal uncertainty in momentum. The nonzero value of $\Delta q_0$ is a
manifestation of the quantization of space. There exists a lower bound to the
possible resolution with which we can measure distances. From \eqref{eq:1} we
can infer that the absolutely smallest uncertainty in position has the value
\begin{equation}
\Delta q_0 = \hbar\sqrt\beta.
\end{equation}
We see that the constant $\beta$ describes the quantization of space where for
$\beta = 0$ we recover the usual uncertainty relation and nonquantized space.
The uncertainty relation \eqref{eq:1} can be derived from the following
commutation relations
\begin{equation}
[\hat{q},\hat{p}] = i\hbar\left(\hat{1} + \beta\hat{p}^2\right)
\label{eq:2}
\end{equation}
for the operators of position and momentum. In literature can be found several
works in which the value of the constant $\beta$ is assessed by means of
different approaches. We mention papers \cite{Scardigli:2015,Scardigli:2017} in
which authors consider generalized uncertainty principle in quantum gravity and
calculate bounds for the parameter $\beta$ from well-known astronomical
measurements as well as compute its value by using quantum corrections to
the Newtonian potential.

It should be noted that in literature can be found approaches to quantum
mechanics and field theory with different modified commutation relations than
the one considered in this work. For instance in \cite{Amelino-Camelia:2002,%
Freidel:2007,Meljanac:2008} authors developed quantum field theory on a
noncommutative space-time characterized by the commutation relations of
space-time variables either of the canonical type
\begin{equation}
[x_\mu,x_\nu] = i\theta_{\mu\nu}
\end{equation}
or the Lie-algebra type
\begin{equation}
[x_\mu,x_\nu] = iC^\beta_{\mu\nu} x_\beta,
\end{equation}
where $\theta_{\mu\nu}$ and $C^\beta_{\mu\nu}$ are real constant tensors.
However, such commutation relations lead to different physical and mathematical
consequences.

The purpose of the paper is to develop a deformation quantization approach to
quantum mechanics which will exhibit a nonzero minimal uncertainty in position.
A deformation quantization is a method of quantizing a classical Hamiltonian
system by means of a suitable deformation of a Poisson algebra of the system.
The deformation is performed with respect to some parameter which in the context
of quantization is taken as the Planck's constant $\hbar$. For a detailed
description of deformation quantization theory we refer to
\cite{Bayen:1978a,Bayen:1978b,Blaszak:2012,Gosson:2006}.

To our knowledge in the literature nothing can be found on the subject of
deformation quantization with minimal length. Worth noting are papers
\cite{Amelino-Camelia:2002,Meljanac:2008} cited earlier where star-products on
a noncommutative space-time were introduced and papers
\cite{Bastos:2008,Bastos:2010,Dias:2010} where authors consider quantum
mechanics from a deformation quantization perspective with the following
modified commutation relations of position and momentum operators
\begin{equation}
[\hat{q}_i,\hat{q}_j] = i\theta_{ij}, \quad
[\hat{q}_i,\hat{p}_j] = i\hbar\delta_{ij}, \quad
[\hat{p}_i,\hat{p}_j] = i\eta_{ij},
\end{equation}
where $\theta_{ij}$ and $\eta_{ij}$ are anti-symmetric real constant matrices.
However, no one investigated deformation quantization approach to quantum
mechanics with commutation relations \eqref{eq:2}, which is the topic of this
paper.

The paper is organized as follows. In Section~\ref{sec:2} we introduce a formal
star-product which incorporates a minimal length scale. This formal star-product
serves as a motivation for introducing, in next sections, a non-formal
deformation quantization with minimal length. The term non-formal means that the
star-product of two complex-valued functions results in a complex-valued
function and not in a function with values in a ring
$\mathbb{C}\lshad\hbar\rshad$ of formal power series in $\hbar$ with
coefficients in $\mathbb{C}$. First, in Section~\ref{sec:3} is developed a
generalized arithmetics on $\mathbb{R}$ which plays a fundamental role in the
presented theory. Next, in Section~\ref{sec:4} we introduce integral formulas
for the star-product and a suitable space of functions on $\mathbb{R}^2$ for
which the introduced integral formulas are well defined. We also prove basic
properties of the star-product. Section~\ref{sec:5} presents an extension of the
star-product to a certain Hilbert subspace of $L^2(\mathbb{R}^2,\dd{l})$, and
in Section~\ref{sec:6} we extend the star-product to a suitable space of
distributions. In Section~\ref{sec:7} we define a \nbr{C^*}algebra of
observables and a space of states. We also give a characterization of states in
terms of quasi-probabilistic distribution functions. Section~\ref{sec:8}
presents a construction of an operator representation in momentum space of the
developed formalism of quantum mechanics. In Sections~\ref{sec:9}
and~\ref{sec:10} is derived a form of quasi-probabilistic distribution functions
describing position eigenvectors and states of maximal localization. We end the
paper with some final remarks and conclusions given in Section~\ref{sec:11}.

\section{Motivation}
\label{sec:2}
The starting point of our considerations is a classical system defined on
a phase space $\mathbb{R}^2$ with a canonical Poisson bracket
\begin{equation}
\{q,p\} = 1.
\end{equation}
We can perform quantization of the above classical system by means of a
deformation quantization methods. For this we have to introduce a star-product
on the phase space $\mathbb{R}^2$, which will be a deformation, with respect to
$\hbar$, of the ordinary point-wise product in the Poisson algebra
$C^\infty(\mathbb{R}^2)$, such that the star-product will remain associative but
will loose commutativity. The first order term in the expansion with respect to
$\hbar$ of the star-commutator should be equal
\begin{equation}
[f,g]_\star = i\hbar \{f,g\} + \dotsb.
\end{equation}
Moreover, the following canonical commutation relation should hold
\begin{equation}
[q,p]_\star = i\hbar.
\end{equation}
The most natural family of star-products on $\mathbb{R}^2$ satisfying the above
conditions is of the form
\begin{equation}
f \star g = f\exp\left(i\hbar(1 - \lambda)\overleftarrow{\partial_{q}}
    \overrightarrow{\partial_{p}}
    - i\hbar\lambda\overleftarrow{\partial_{p}}
    \overrightarrow{\partial_{q}}\right)g,
\label{eq:3}
\end{equation}
where $0 \leq \lambda \leq 1$ describes different orderings of position
and momentum operators in a corresponding operator representation.

We can now incorporate a minimal length scale into the above picture by
deformation of the \nbr{\star}product \eqref{eq:3} with respect to $\beta$ to a
new product satisfying, on account of \eqref{eq:2}, the following relation
\begin{equation}
[q,p]_\star = i\hbar(1 + \beta p \star p).
\label{eq:43}
\end{equation}
One of such deformations which has particularly simple form is given by
\begin{align}
f \star g & = f\exp\left(
    i\hbar(1 - \lambda)\overleftarrow{\partial_q}\overrightarrow{D_p}
    - i\hbar\lambda\overleftarrow{D_p}\overrightarrow{\partial_q}\right)g
    \nonumber \\
& = \sum_{k=0}^\infty \frac{1}{k!} (i\hbar)^k \sum_{l=0}^k \binom{k}{l}
    (1 - \lambda)^l (-\lambda)^{k-l} (\partial_q^l D_p^{k-l} f)
    (\partial_q^{k-l} D_p^l g),
\label{eq:4}
\end{align}
where $D_p = (1 + \beta p^2)\partial_p$. Note, that
\begin{equation}
q \star q = q^2, \quad p \star p = p^2,
\label{eq:42}
\end{equation}
where $q^2$ and $p^2$ are the usual products of $q$ and $p$ with themselves.
Moreover, after performing the following noncanonical transformation of
coordinates
\begin{equation}
\begin{aligned}
\bar{q} & = q, \\
\bar{p} & = \frac{1}{\sqrt\beta}\arctan(\sqrt\beta p)
\end{aligned}
\end{equation}
the \nbr{\star}product \eqref{eq:4} takes the form as in \eqref{eq:3}:
\begin{equation}
f \star g = f\exp\left(i\hbar(1 - \lambda)\overleftarrow{\partial_{\bar{q}}}
    \overrightarrow{\partial_{\bar{p}}}
    - i\hbar\lambda\overleftarrow{\partial_{\bar{p}}}
    \overrightarrow{\partial_{\bar{q}}}\right)g,
\end{equation}
and the star-commutator of $\bar{q}$, $\bar{p}$ is equal $i\hbar$.

In the limit $\hbar \to 0$ the \nbr{\star}product \eqref{eq:4} reduces to the
ordinary point-wise product and $\frac{1}{i\hbar}[\sdot,\sdot]_\star$ reduces to
a Poisson bracket which acts on observables of position and momentum through the
following relation
\begin{equation}
\{q,p\} = 1 + \beta p^2.
\end{equation}
Thus in the limit $\hbar \to 0$ we received a classical system with a
noncanonical Poisson bracket. Investigations of this classical system could
reveal something about the parameter $\beta$. For example in \cite{Tkachuk:2012}
author argued that in order to preserve the correspondence principle the
parameter $\beta$ has to depend on the mass of a particle which dynamics is
described by the system.

In what follows we will investigate properties of the star-product \eqref{eq:4}
and develop a complete theory of quantum mechanics on phase space exhibiting a
nonzero minimal uncertainty in position. As we will see the presented
quantization procedure is connected with a modified arithmetics on $\mathbb{R}$,
which we will present in the next section.

\section{Generalized arithmetics on $\mathbb{R}$}
\label{sec:3}
In what follows we will develop a modified arithmetic on $\mathbb{R}$. It will
be a particular case of a general approach considered in \cite{Czachor:2016}.

Let $\beta > 0$. For $x,y \in \mathbb{R}$ such that $\beta xy \neq 1$ we define
the generalized addition by
\begin{equation}
x \oplus y = \frac{x + y}{1 - \beta xy}.
\end{equation}
The set $\{(x,y) \in \mathbb{R}^2 \mid \beta xy = 1\}$ on which the operation
$\oplus$ is not well defined is of Lebesgue measure zero. The generalized
addition $\oplus$ has the following properties:
\begin{enumerate}[(i)]
\item\label{item:1a} $x \oplus (y \oplus z) = (x \oplus y) \oplus z$
(associativity),
\item\label{item:1b} $x \oplus y = y \oplus x$ (commutativity),
\item\label{item:1c} $x \oplus 0 = 0 \oplus x = x$ (0 is the neutral element),
\item\label{item:1d} $x \oplus (-x) = (-x) \oplus x = 0$ ($-x$ is the inverse
element to $x$),
\item\label{item:1e} $x \oplus y = \frac{1}{\sqrt\beta}\tan\left(
\arctan(\sqrt\beta x) + \arctan(\sqrt\beta y)\right)$.
\end{enumerate}
Indeed, we calculate that
\begin{equation}
x \oplus (y \oplus z) =
    \frac{x + \frac{y + z}{1 - \beta yz}}{1 - \beta x\frac{y + z}{1 - \beta yz}}
= \frac{x + y + z - \beta xyz}{1 - \beta(xy + yz + zx)}
= (x \oplus y) \oplus z,
\end{equation}
which proves property~(\ref{item:1a}). Properties (\ref{item:1b}),
(\ref{item:1c}), and (\ref{item:1d}) are an immediate consequence of the
definition. Property (\ref{item:1e}) follows from the formula for the tangent
of the sum of angles:
\begin{equation}
\tan(\theta_1 + \theta_2) =
    \frac{\tan\theta_1 + \tan\theta_2}{1 - \tan\theta_1 \tan\theta_2}.
\end{equation}
For every $y \in \mathbb{R}$ such that $y \neq 0$ the map
$x \mapsto x \oplus y$ is a bijection of the set $\{x \in \mathbb{R} \mid
x \neq \tfrac{1}{\beta y}\}$ onto the set $\{z \in \mathbb{R} \mid
z \neq -\tfrac{1}{\beta y}\}$. We define the generalized subtraction by
\begin{equation}
x \ominus y = x \oplus (-y).
\end{equation}

For a scalar $\lambda \in \mathbb{R}$ such that $-1 \leq \lambda \leq 1$ and
a vector $x \in \mathbb{R}$ we define the generalized multiplication by scalar:
\begin{equation}
\lambda \circ x =
    \frac{1}{\sqrt\beta}\tan\left(\lambda\arctan(\sqrt\beta x)\right).
\end{equation}
Note, that for a fixed $\lambda \in [-1,1]$ the map $x \mapsto \lambda
\circ x$ is a bijection of $\mathbb{R}$ into $\mathbb{R}$ and for a fixed
$x \in \mathbb{R}$ the map $\lambda \mapsto \lambda \circ x$ is a
bijection of $[-1,1]$ into $\mathbb{R}$. The operation $\circ$ has the
following properties:
\begin{enumerate}[(i)]
\item\label{item:2a} $\lambda \circ (x \oplus y) =
(\lambda \circ x) \oplus (\lambda \circ y)$,
\item\label{item:2b} $(\lambda_1 \circ x) \oplus
(\lambda_2 \circ x) = (\lambda_1 + \lambda_2) \circ x$,
\item\label{item:2c} $\lambda_1 \circ (\lambda_2 \circ x) =
(\lambda_1 \lambda_2) \circ x$.
\end{enumerate}
Indeed, we calculate that
\begin{align}
(\lambda \circ x) \oplus (\lambda \circ y) & =
    \frac{1}{\sqrt\beta}\tan\left(
    \arctan\bigl(\sqrt\beta (\lambda \circ x)\bigr)
    + \arctan\bigl(\sqrt\beta (\lambda \circ y)\bigr)\right) \nonumber \\
& = \frac{1}{\sqrt\beta}\tan\left(\lambda\arctan(\sqrt\beta x)
    + \lambda\arctan(\sqrt\beta y)\right)
= \frac{1}{\sqrt\beta}\tan\left(
    \lambda\arctan\bigl(\sqrt\beta(x \oplus y)\bigr)\right) \nonumber \\
& = \lambda \circ (x \oplus y),
\end{align}
which proves (\ref{item:2a}). Property (\ref{item:2b}) follows from
\begin{align}
(\lambda_1 \circ x) \oplus (\lambda_2 \circ x) & =
    \frac{1}{\sqrt\beta}\tan\left(
    \arctan\bigl(\sqrt\beta (\lambda_1 \circ x)\bigr)
    + \arctan\bigl(\sqrt\beta (\lambda_2 \circ x)\bigr)\right)
    \nonumber \\
& = \frac{1}{\sqrt\beta}\tan\left((\lambda_1 + \lambda_2)\arctan(\sqrt\beta x)
    \right)
= (\lambda_1 + \lambda_2) \circ x,
\end{align}
and property (\ref{item:2c}) is a consequence of
\begin{align}
\lambda_1 \circ (\lambda_2 \circ x) & =
    \frac{1}{\sqrt\beta}\tan\left(\lambda_1\arctan\left(\sqrt\beta
    \frac{1}{\sqrt\beta}\tan\left(\lambda_2\arctan(\sqrt\beta x)\right)\right)
    \right) \nonumber \\
& = \frac{1}{\sqrt\beta}\tan\left(\lambda_1\lambda_2\arctan(\sqrt\beta x)\right)
= (\lambda_1 \lambda_2) \circ x.
\end{align}
From the half-angle formula for the tangent function
\begin{equation}
\tan\frac{\theta}{2} = \frac{\sqrt{1 + \tan^2\theta} - 1}{\tan\theta}, \quad
\theta \in (-\tfrac{\pi}{2},\tfrac{\pi}{2})
\end{equation}
we get that in particular
\begin{equation}
\frac{1}{2} \circ x = \frac{\sqrt{1 + \beta x^2} - 1}{\beta x}.
\end{equation}

The definition of the generalized addition $\oplus$ can be continuously extended
to the projectively extended space of real numbers
$\overline{\mathbb{R}} = \mathbb{R} \cup \{\infty\}$, so that it will be well
defined for every pair of points $x,y \in \overline{\mathbb{R}}$, if we put
\begin{equation}
\begin{aligned}
x \oplus \frac{1}{\beta x} & = \frac{1}{\beta x} \oplus x = \infty
\text{ for $x \in \mathbb{R}$ such that $x \neq 0$}, \\
x \oplus \infty & = \infty \oplus x = -\frac{1}{\beta x}
\text{ for $x \in \mathbb{R}$ such that $x \neq 0$}, \\
0 \oplus \infty & = \infty \oplus 0 = \infty, \\
\infty \oplus \infty & = 0.
\end{aligned}
\end{equation}
Then, the space $\overline{\mathbb{R}}$ together with the extended operation
$\oplus$ becomes a compact abelian topological group isomorphic to the circle
group $\mathbb{T}^1$. For example if we define $\mathbb{T}^1$ as the quotient
group $\mathbb{R}/2\pi\mathbb{Z}$ equipped with the usual addition of numbers
modulo $2\pi$, then the map
\begin{equation}
\Phi(x) = \frac{1}{\sqrt\beta}\tan\frac{x}{2}
\label{eq:9}
\end{equation}
is an isomorphism of $\mathbb{T}^1$ onto $\overline{\mathbb{R}}$.

For a function $f \colon \mathbb{R} \to \mathbb{C}$ we may define its
generalized derivative by
\begin{equation}
D_p f(p) = \lim_{\eta \to 0} \frac{f(p \oplus \eta) - f(p)}{\eta}.
\end{equation}
There holds
\begin{equation}
D_p f(p) = (1 + \beta p^2)\partial_p f(p).
\end{equation}
Indeed, introducing $L_p(\eta) = p \oplus \eta =
\frac{p + \eta}{1 - \beta p\eta}$ we calculate that
\begin{equation}
L'_p(\eta) = \frac{1 + \beta p^2}{(1 - \beta p\eta)^2}
\end{equation}
and
\begin{align}
D_p f(p) & = \lim_{\eta \to 0} \frac{f(p \oplus \eta) - f(p)}{\eta}
= \lim_{\eta \to 0} \frac{f(L_p(\eta)) - f(L_p(0))}{\eta}
= \frac{\dd}{\dd\eta} f(L_p(\eta)) \bigg|_{\eta = 0} = f'(L_p(0)) L'_p(0)
    \nonumber \\
& = (1 + \beta p^2)\partial_p f(p).
\end{align}

The following identities hold
\begin{equation}
\label{eq:10}
\begin{gathered}
D_p f(p \oplus \eta) = D_\eta f(p \oplus \eta) = (D_p f)(p \oplus \eta), \\
\lambda D_p f(p \oplus \lambda \circ \eta) =
    D_\eta f(p \oplus \lambda \circ \eta)
= \lambda (D_p f)(p \oplus \lambda \circ \eta).
\end{gathered}
\end{equation}
Indeed, introducing the notation $h_\lambda(\eta) = \lambda \circ \eta$ we
have that $h'_\lambda(\eta) =
\lambda \frac{1 + \beta h_\lambda^2(\eta)}{1 + \beta \eta^2}$ and
\begin{equation}
\begin{aligned}
D_\eta f(p \oplus \lambda \circ \eta) & = f'(p \oplus \lambda \circ \eta)
    D_\eta(p \oplus \lambda \circ \eta)
= f'(p \oplus \lambda \circ \eta)
    \frac{(1 + \beta p^2)(1 + \beta \eta^2)}{(1 - \beta p h_\lambda(\eta))^2}
    h'_\lambda(\eta) \\
& = \lambda f'(p \oplus \lambda \circ \eta)
    \frac{(1 + \beta p^2)(1 + \beta h_\lambda^2(\eta))}
    {(1 - \beta p h_\lambda(\eta))^2}
= \lambda D_p f(p \oplus \lambda \circ \eta), \\
(D_p f)(p \oplus \lambda \circ \eta) & =
    \bigl(1 + \beta(p \oplus \lambda \circ \eta)^2\bigr)
    f'(p \oplus \lambda \circ \eta)
= \frac{(1 + \beta p^2)(1 + \beta h_\lambda^2(\eta))}
    {(1 - \beta p h_\lambda(\eta))^2} f'(p \oplus \lambda \circ \eta) \\
& = D_p f(p \oplus \lambda \circ \eta).
\end{aligned}
\end{equation}

For $q,p \in \mathbb{R}$ we will use the following notation
\begin{equation}
(q,p) = q\frac{1}{\sqrt\beta}\arctan(\sqrt\beta p).
\end{equation}
We easily check that
\begin{equation}
\begin{aligned}
(q_1 + q_2,p) & = (q_1,p) + (q_2,p), \\
(q,p_1 \oplus p_2) & = (q,p_1) + (q,p_2), \\
\lambda (q,p) & = (\lambda q,p) = (q,\lambda \circ p).
\end{aligned}
\end{equation}
The operation $(q,p)$ can be treated as a generalized scalar product of $q$ and
$p$.

On $\mathbb{R}$ we can introduce a measure which will be invariant with respect
to generalized translations and which will reduce to the Lebesgue measure for
$\beta = 0$. The proper measure is equal
\begin{equation}
\dd{\mu(p)} = \frac{\dd p}{1 + \beta p^2}.
\label{eq:30}
\end{equation}
One easily checks that indeed for $f \in L^1(\mathbb{R})$ and
$\eta \in \mathbb{R}$
\begin{equation}
\int_\mathbb{R} f(p \oplus \eta) \ud{\mu(p)} = \int_\mathbb{R} f(p) \ud{\mu(p)}.
\end{equation}

On the phase space $\mathbb{R}^2$ we introduce the following measure
\begin{equation}
\dd{l(x)} = \frac{1}{2\pi\hbar} \ud{q}\ud{\mu(p)},
\end{equation}
where $x = (q,p)$. Let us denote by $\mathcal{L}$ the space of square integrable
functions $f \in L^2(\mathbb{R}^2,\dd{l})$ such that for almost every
$p \in \mathbb{R}$ the functions $q \mapsto f(q,p)$ extend to entire functions
on $\mathbb{C}$ of exponential type $\leq \tfrac{\pi}{2\sqrt\beta}$, i.e.
$f \in \mathcal{L}$ if and only if $f \in L^2(\mathbb{R}^2,\dd{l})$ and for
almost every $p \in \mathbb{R}$ the function $q \mapsto f(q,p)$ extend to entire
function on $\mathbb{C}$ satisfying for some constant $C(p)$ the following
condition
\begin{equation}
\abs{f(z,p)} \leq C(p) e^{\frac{\pi}{2\sqrt\beta}\abs{z}}, \quad
z \in \mathbb{C}.
\end{equation}
The space $\mathcal{L}$ is a Hilbert subspace of $L^2(\mathbb{R}^2,\dd{l})$.

For $f \in \mathcal{L}$ we define its generalized Fourier transform in
position variable by
\begin{equation}
\tilde{f}(p',p) = \int_\mathbb{R} f(q,p) e^{-\frac{i}{\hbar}(q,p')} \ud{q}.
\end{equation}
If $f$ is not Lebesgue integrable, then the integral has to be considered as
the improper integral
\begin{equation}
\lim_{R \to \infty} \int_{-R}^R \dotsc \ud{q}.
\label{eq:40}
\end{equation}
This transform is a Hilbert space isomorphism of $\mathcal{L}$ onto
$L^2(\mathbb{R}^2,\dd{\mu} \times \dd{\mu})$ and its inverse is given by
\begin{equation}
f(q,p) = \frac{1}{2\pi\hbar} \int_\mathbb{R} \tilde{f}(p',p)
    e^{\frac{i}{\hbar}(q,p')} \ud{\mu(p')}.
\end{equation}
Indeed, if we denote by $\check{f}$ the usual Fourier transform of $f$ in
position variable, then
\begin{align}
& \frac{1}{2\pi\hbar} \int_\mathbb{R} \tilde{f}(p',p)
    e^{\frac{i}{\hbar}(q,p')} \ud{\mu(p')} =
    \frac{1}{2\pi\hbar} \int_\mathbb{R}
    \check{f}\left(\frac{1}{\sqrt\beta}\arctan(\sqrt\beta p'),p\right)
    e^{\frac{i}{\hbar}(q,p')} \ud{\mu(p')} \nonumber \\
& \quad = \frac{1}{2\pi\hbar} \int_\mathbb{R}
    \check{f}\left(\frac{1}{\sqrt\beta}\arctan(\sqrt\beta p'),p\right)
    \exp\left(\frac{i}{\hbar}q\frac{1}{\sqrt\beta}\arctan(\sqrt\beta p')\right)
    \frac{1}{1 + \beta p'^2} \ud{p'} \nonumber \\
& \quad = \frac{1}{2\pi\hbar}
    \int_{-\frac{\pi}{2\sqrt\beta}}^{+\frac{\pi}{2\sqrt\beta}}
    \check{f}(\bar{p},p)
    e^{\frac{i}{\hbar}q\bar{p}} \ud{\bar{p}}.
\end{align}
By Paley-Wiener theorem $\bar{p} \mapsto \check{f}(\bar{p},p)$ is square
integrable with respect to the Lebesgue measure $\dd{\bar{p}}$ and its support
lies in $\bigl[-\tfrac{\pi}{2\sqrt\beta},\tfrac{\pi}{2\sqrt\beta}\bigr]$. Thus
we can extend the integration with respect to $\bar{p}$ from $-\infty$ to
$\infty$, which yields the inverse Fourier transform in position variable equal
to $f(q,p)$ almost everywhere.

The generalized symplectic Fourier transform of a function
$f \in \mathcal{L}$ will be denoted by $\mathcal{F}_\beta f$ and
defined by the formula
\begin{equation}
\mathcal{F}_\beta f(q',p') = \frac{1}{2\pi\hbar} \int_{\mathbb{R}^2} f(q,p)
    e^{-\frac{i}{\hbar}(q,p')} e^{\frac{i}{\hbar}(q',p)} \ud{q}\ud{\mu(p)}.
\end{equation}
It is its own inverse: $\mathcal{F}_\beta^{-1} = \mathcal{F}_\beta$.
The generalized symplectic Fourier transform has the following properties
\begin{subequations}
\label{eq:5}
\begin{gather}
D_p \mathcal{F}_\beta f = -\frac{i}{\hbar} \mathcal{F}_\beta (qf), \quad
\partial_q \mathcal{F}_\beta f = \frac{i}{\hbar}
    \mathcal{F}_\beta \left(\frac{1}{\sqrt\beta}\arctan(\sqrt\beta p)f\right),
    \label{eq:5a} \\
q \mathcal{F}_\beta f = i\hbar \mathcal{F}_\beta (D_p f), \quad
\frac{1}{\sqrt\beta}\arctan(\sqrt\beta p) \mathcal{F}_\beta f =
    -i\hbar \mathcal{F}_\beta (\partial_q f), \label{eq:5b} \\
\mathcal{F}_\beta (f \cdot g) =
    \frac{1}{2\pi\hbar} \mathcal{F}_\beta f \circledast \mathcal{F}_\beta g,
    \label{eq:5c}
\end{gather}
\end{subequations}
where
\begin{equation}
(f \circledast g)(q,p) = \int_{\mathbb{R}^2} f(q',p')
    g(q - q',p \ominus p') \ud{q'}\ud{\mu(p')}
= \int_{\mathbb{R}^2} f(q - q',p \ominus p') g(q',p') \ud{q'}\ud{\mu(p')}
\label{eq:6}
\end{equation}
is a generalized convolution.

\section{Properties of the $\star$-product}
\label{sec:4}
Based on the formal star-product \eqref{eq:4} we can get an integral formula for
the star-product, which will describe a well defined non-formal star-product on
a certain space of functions. In the rest of the paper we will develop a
non-formal deformation quantization theory based on this integral formula for
the star-product. The equation \eqref{eq:4} will represent a formal expansion of
this star-product in $\hbar$.

\begin{theorem}
The star-product \eqref{eq:4} can be written in the following integral forms
\begin{subequations}
\label{eq:7}
\begin{align}
(f \star g)(q,p) & = \frac{1}{(2\pi\hbar)^2} \int_{\mathbb{R}^4}
    \mathcal{F}_\beta f(q',p') \mathcal{F}_\beta g(q'',p'')
    e^{\frac{i}{\hbar}(q - \lambda q',p'')}
    e^{-\frac{i}{\hbar}(q'',p \ominus (1 - \lambda) \circ p')}
    e^{\frac{i}{\hbar}(q,p')} e^{-\frac{i}{\hbar}(q',p)} \nonumber \\
& \quad \times \ud{q'}\ud{q''}\ud{\mu(p')}\ud{\mu(p'')} \label{eq:7a} \\
& = \frac{1}{(2\pi\hbar)^2} \int_{\mathbb{R}^2}
    \tilde{f}\bigl(p',p \oplus \lambda \circ p''\bigr)
    \tilde{g}\bigl(p'',p \ominus (1 - \lambda) \circ p'\bigr)
    e^{\frac{i}{\hbar}(q,p' \oplus p'')} \ud{\mu(p')}\ud{\mu(p'')}
    \label{eq:7b} \\
& = \frac{1}{(2\pi\hbar)^2} \int_{\mathbb{R}^4}
    f\bigl(q + q',p \oplus \lambda \circ p''\bigr)
    g\bigl(q + q'',p \ominus (1 - \lambda) \circ p'\bigr)
    e^{-\frac{i}{\hbar}(q'',p'')} e^{-\frac{i}{\hbar}(q',p')} \nonumber \\
& \quad \times \ud{q'}\ud{q''}\ud{\mu(p')}\ud{\mu(p'')}. \label{eq:7c}
\end{align}
\end{subequations}
\end{theorem}

\begin{proof}
Taking the generalized symplectic Fourier transform of \eqref{eq:4} and using
\eqref{eq:5b}, \eqref{eq:5c}, \eqref{eq:6} we find that
\begin{align}
& \mathcal{F}_\beta(f \star g)(q'',p'') = \frac{1}{2\pi\hbar} \sum_{k=0}^\infty
    \frac{1}{k!} (i\hbar)^k \sum_{l=0}^k \binom{k}{l} (1 - \lambda)^l
    (-\lambda)^{k-l} \int_{\mathbb{R}^2}
    \left(\frac{i}{\hbar} \frac{1}{\sqrt\beta}\arctan(\sqrt\beta p')\right)^l
    \left(-\frac{i}{\hbar} q'\right)^{k-l} \nonumber \\
& \qquad \times \mathcal{F}_\beta f(q',p') \left(\frac{i}{\hbar}
    \frac{1}{\sqrt\beta}\arctan\bigl(\sqrt\beta
    (p'' \ominus p')\bigr)\right)^{k-l}
    \left(-\frac{i}{\hbar} (q'' - q')\right)^l
    \mathcal{F}_\beta g(q'' - q',p'' \ominus p')
    \ud{q'}\ud{\mu(p')} \nonumber \\
& \quad = \frac{1}{2\pi\hbar} \int_{\mathbb{R}^2} \sum_{k=0}^\infty \frac{1}{k!}
    \left(\frac{i}{\hbar}\right)^k \sum_{l=0}^k \binom{k}{l}
    \bigl((1 - \lambda)(q'' - q',p')\bigr)^l
    \bigl(-\lambda(q',p'' \ominus p')\bigr)^{k-l}
    \mathcal{F}_\beta f(q',p') \nonumber \\
& \qquad \times \mathcal{F}_\beta g(q'' - q',p'' \ominus p')
    \ud{q'}\ud{\mu(p')} \nonumber \\
& \quad = \frac{1}{2\pi\hbar} \int_{\mathbb{R}^2} \mathcal{F}_\beta f(q',p')
    \mathcal{F}_\beta g(q'' - q',p'' \ominus p')
    e^{\frac{i}{\hbar}(1 - \lambda)(q'' - q',p')}
    e^{-\frac{i}{\hbar}\lambda(q',p'' \ominus p')}
    \ud{q'}\ud{\mu(p')} \label{eq:8}
\end{align}
Taking the generalized symplectic Fourier transform of \eqref{eq:8} and
performing the following change of variables under the integral sign
\begin{equation}
\begin{aligned}
\tilde{p}'' & = p'' \ominus p', & \tilde{q}'' & = q'' - q', \\
\tilde{p}' & = p', & \tilde{q}' & = q',
\end{aligned}
\end{equation}
gives \eqref{eq:7a}. By performing integration with respect to $q''$ and $q'$
in \eqref{eq:7a} we get \eqref{eq:7b}, and by expanding Fourier transforms in
\eqref{eq:7b} we receive \eqref{eq:7c}.
\end{proof}

Note, that by defining a twisted convolution $\diamond$ of functions $f$ and $g$
by the relation
\begin{equation}
\mathcal{F}_\beta(f \star g) = \frac{1}{2\pi\hbar}
    \mathcal{F}_\beta f \diamond \mathcal{F}_\beta g
\label{eq:15}
\end{equation}
we get by virtue of \eqref{eq:8} that
\begin{equation}
(f \diamond g)(q,p) = \int_{\mathbb{R}^2} f(q',p') g(q - q',p \ominus p')
    e^{\frac{i}{\hbar}(1 - \lambda)(q - q',p')}
    e^{-\frac{i}{\hbar}\lambda(q',p \ominus p')}
    \ud{q'}\ud{\mu(p')}.
\label{eq:20}
\end{equation}
Moreover, we can also define a twisted convolution with respect to position
variables by the formula
\begin{equation}
\tilde{f} \odot \tilde{g} = (f \star g)^\sim.
\end{equation}
Then, we find from \eqref{eq:7b} that
\begin{equation}
(\tilde{f} \odot \tilde{g})(p',p) = \frac{1}{2\pi\hbar} \int_\mathbb{R}
    \tilde{f}\bigl(p'',p \oplus \lambda \circ (p' \ominus p'')\bigr)
    \tilde{g}\bigl(p' \ominus p'',p \ominus (1 - \lambda) \circ p''\bigr)
    \ud{\mu(p'')}.
\label{eq:11}
\end{equation}

Let us denote by $\widetilde{\mathcal{F}}(\mathbb{R}^2)$ the space of
complex-valued functions on $\mathbb{R}^2$ which extend to smooth functions on
$\overline{\mathbb{R}}^2$. The space $\widetilde{\mathcal{F}}(\mathbb{R}^2)$ is
a Fr\'echet space with topology given by semi-norms
\begin{equation}
\norm{f}_{n,m} = \sup_{(p',p) \in \mathbb{R}^2}\abs{D_{p'}^n D_p^m f(p',p)},
\quad n,m \in \mathbb{N}.
\end{equation}
Indeed, if $\Phi$ is the isomorphism \eqref{eq:9} of the circle $\mathbb{T}^1$
onto $\overline{\mathbb{R}}$ and if we associate with every
$f \in \widetilde{\mathcal{F}}(\mathbb{R}^2)$ a function
$F \in C^\infty(\mathbb{T}^2)$ given by $F(p',p) = f(\Phi(p'),\Phi(p))$, then we
can see that the semi-norms $\norm{\sdot}_{n,m}$ can be written in a form
\begin{equation}
\norm{f}_{n,m} = \bigl(2\sqrt\beta\bigr)^{n+m} \sup_{(p',p) \in \mathbb{T}^2}
    \abs{\partial_{p'}^n \partial_p^m F(p',p)}.
\end{equation}
Thus, the map $f \mapsto F$ is a topological isomorphism of the space
$\widetilde{\mathcal{F}}(\mathbb{R}^2)$ onto the space of smooth functions on
the torus $\mathbb{T}^2 = \mathbb{T}^1 \times \mathbb{T}^1$ and endowed with
a standard Fr\'echet topology of uniform convergence together with all
derivatives.

We will denote by $\mathcal{F}(\mathbb{R}^2)$ the image of the space
$\widetilde{\mathcal{F}}(\mathbb{R}^2)$ with respect to the inverse generalized
Fourier transform in position variable. The space $\mathcal{F}(\mathbb{R}^2)$
carries the Fr\'echet topology induced from
$\widetilde{\mathcal{F}}(\mathbb{R}^2)$.

\begin{theorem}
If $f,g \in \mathcal{F}(\mathbb{R}^2)$, then
$f \star g \in \mathcal{F}(\mathbb{R}^2)$ and $(f,g) \mapsto f \star g$ is a
continuous bilinear operation on $\mathcal{F}(\mathbb{R}^2)$.
\end{theorem}

\begin{proof}
Using \eqref{eq:10} and \eqref{eq:11} we find that
\begin{equation}
\label{eq:12}
\begin{aligned}
D_{p'}(\tilde{f} \odot \tilde{g}) & = \lambda D_p \tilde{f} \odot \tilde{g}
    + \tilde{f} \odot D_{p'} \tilde{g}
= D_{p'} \tilde{f} \odot \tilde{g}
    - (1 - \lambda) \tilde{f} \odot D_p \tilde{g}, \\
D_p (\tilde{f} \odot \tilde{g}) & = D_p \tilde{f} \odot \tilde{g}
    + \tilde{f} \odot D_p \tilde{g}.
\end{aligned}
\end{equation}
Hence, by induction on these formulas if
$\tilde{f},\tilde{g} \in \widetilde{\mathcal{F}}(\mathbb{R}^2)$ then
also $\tilde{f} \odot \tilde{g} \in \widetilde{\mathcal{F}}(\mathbb{R}^2)$.
Thus $\mathcal{F}(\mathbb{R}^2)$ is closed with respect to the
\nbr{\star}product.

From \eqref{eq:12} we get that
\begin{equation}
D_{p'}^n D_p^m (\tilde{f} \odot \tilde{g}) = \sum_{k=0}^n \sum_{l=0}^m
    \binom{n}{k} \binom{m}{l} \lambda^k D_p^{k+l} \tilde{f} \odot D_{p'}^{n-k}
    D_p^{m-l} \tilde{g}.
\end{equation}
and from \eqref{eq:11}
\begin{equation}
\norm{\tilde{f} \odot \tilde{g}}_\infty \leq \frac{1}{2\hbar\sqrt\beta}
    \norm{\tilde{f}}_\infty \norm{\tilde{g}}_\infty.
\end{equation}
Therefore
\begin{equation}
\norm{\tilde{f} \odot \tilde{g}}_{n,m} \leq \frac{1}{2\hbar\sqrt\beta}
    \sum_{k=0}^n \sum_{l=0}^m
    \binom{n}{k} \binom{m}{l} \lambda^k
    \norm{\tilde{f}}_{0,k+l} \norm{\tilde{g}}_{n-k,m-l}
\end{equation}
which shows that $(\tilde{f},\tilde{g}) \mapsto \tilde{f} \odot \tilde{g}$ is
jointly continuous for the topology of $\widetilde{\mathcal{F}}(\mathbb{R}^2)$.
Thus $(f,g) \mapsto f \star g$ is jointly continuous on
$\mathcal{F}(\mathbb{R}^2) \times \mathcal{F}(\mathbb{R}^2)$.
\end{proof}

Note, that the generalized symplectic Fourier transform $\mathcal{F}_\beta$ is
a topological isomorphism of $\mathcal{F}(\mathbb{R}^2)$ onto
$\mathcal{F}(\mathbb{R}^2)$. Thus, by \eqref{eq:15} the space
$\mathcal{F}(\mathbb{R}^2)$ is also closed under the twisted convolution
$\diamond$.

\begin{theorem}
The \nbr{\star}product is associative, i.e.
\begin{equation}
f \star (g \star h) = (f \star g) \star h
\end{equation}
for $f,g,h \in \mathcal{F}(\mathbb{R}^2)$. 
\end{theorem}

\begin{proof}
We have that
\begin{align}
(f \diamond (g \diamond h))(q,p) & = \int_{\mathbb{R}^4} f(q',p') g(q'',p'')
    h(q - q' - q'',p \ominus p' \ominus p'')
    e^{\frac{i}{\hbar}(1 - \lambda)(q - q',p')}
    e^{-\frac{i}{\hbar}\lambda(q',p \ominus p')} \nonumber \\
& \quad \times e^{\frac{i}{\hbar}(1 - \lambda)(q - q' - q'',p'')}
    e^{-\frac{i}{\hbar}\lambda(q'',p \ominus p' \ominus p'')}
    \ud{q'}\ud{\mu(p')}\ud{q''}\ud{\mu(p'')} \nonumber \\
& = \int_{\mathbb{R}^4} f(q',p') g(q'' - q',p'' \ominus p')
    h(q - q'',p \ominus p'')
    e^{\frac{i}{\hbar}(1 - \lambda)(q'' - q',p')}
    e^{-\frac{i}{\hbar}\lambda(q',p'' \ominus p')} \nonumber \\
& \quad \times e^{\frac{i}{\hbar}(1 - \lambda)(q - q'',p'')}
    e^{-\frac{i}{\hbar}\lambda(q'',p \ominus p'')}
    \ud{q'}\ud{\mu(p')}\ud{q''}\ud{\mu(p'')} \nonumber \\
& = ((f \diamond g) \diamond h)(q,p).
\end{align}
Applying \eqref{eq:15} yields the associativity of $\star$.
\end{proof}

\begin{theorem}
For $f,g \in \mathcal{F}(\mathbb{R}^2)$
\begin{equation}
\int_{\mathbb{R}^2} (f \star g)(q,p) \ud{q}\ud{\mu(p)} =
    \int_{\mathbb{R}^2} (g \star f)(q,p) \ud{q}\ud{\mu(p)}.
\label{eq:13}
\end{equation}
In particular, for $\lambda = \tfrac{1}{2}$
\begin{equation}
\int_{\mathbb{R}^2} (f \star g)(q,p) \ud{q}\ud{\mu(p)} =
    \int_{\mathbb{R}^2} f(q,p)g(q,p) \ud{q}\ud{\mu(p)}.
\label{eq:14}
\end{equation}
\end{theorem}

\begin{proof}
We have that
\begin{align}
\int_{\mathbb{R}^2} (f \star g)(q,p) \ud{q}\ud{\mu(p)} & =
    2\pi\hbar \mathcal{F}_\beta(f \star g)(0)
= (\mathcal{F}_\beta f \diamond \mathcal{F}_\beta g)(0) \nonumber \\
& = \int_{\mathbb{R}^2} \mathcal{F}_\beta f(q',p') \mathcal{F}_\beta g(-q',-p')
    e^{-\frac{i}{\hbar}(1 - 2\lambda)(q',p')} \ud{q'}\ud{\mu(p')}.
\label{eq:21}
\end{align}
Performing the following change of variables under the integral sign:
$q'' = -q'$, $p'' = -p'$ we get \eqref{eq:13}. If $\lambda = \tfrac{1}{2}$, then
\begin{align}
\int_{\mathbb{R}^2} (f \star g)(q,p) \ud{q}\ud{\mu(p)} & =
    \int_{\mathbb{R}^2} \mathcal{F}_\beta f(q',p') \mathcal{F}_\beta g(-q',-p')
    \ud{q'}\ud{\mu(p')}
= (\mathcal{F}_\beta f \circledast \mathcal{F}_\beta g)(0)
= 2\pi\hbar \mathcal{F}_\beta(f \cdot g)(0) \nonumber \\
& = \int_{\mathbb{R}^2} f(q,p)g(q,p) \ud{q}\ud{\mu(p)},
\end{align}
which proves \eqref{eq:14}.
\end{proof}

\begin{theorem}
For $f,g \in \mathcal{F}(\mathbb{R}^2)$
\begin{equation}
\partial_q(f \star g) = \partial_q f \star g + f \star \partial_q g, \quad
D_p(f \star g) = D_p f \star g + f \star D_p g.
\end{equation}
\end{theorem}

\begin{proof}
These equalities follow immediately from \eqref{eq:7c} and \eqref{eq:10}.
\end{proof}

For $f \in \mathcal{F}(\mathbb{R}^2)$ we define its conjugation $f^*$ by the
formula
\begin{equation}
f^*(q,p) = \frac{1}{2\pi\hbar} \int_{\mathbb{R}^2} \mathcal{F}_\beta
    \bar{f}(q',p') e^{\frac{i}{\hbar}(1 - 2\lambda)(q',p')}
    e^{\frac{i}{\hbar}(q,p')} e^{-\frac{i}{\hbar}(q',p)} \ud{q'}\ud{\mu(p')},
\label{eq:16}
\end{equation}
where $\bar{f}$ denotes the complex-conjugation of $f$. For
$\lambda = \tfrac{1}{2}$ it can be seen from \eqref{eq:16} that the conjugation
$*$ is the usual complex-conjugation. Formula \eqref{eq:16} can be also written
in the form
\begin{equation}
f^*(q,p) = \frac{1}{2\pi\hbar} \int_{\mathbb{R}}
    \overline{\tilde{f}\bigl(-p',p \ominus (1 - 2\lambda) \circ p'\bigr)}
    e^{\frac{i}{\hbar}(q,p')} \ud{\mu(p')}.
\label{eq:17}
\end{equation}
From \eqref{eq:17} we get that
\begin{equation}
\widetilde{f^*}(p',p) = \overline{\tilde{f}\bigl(-p',p \ominus (1 - 2\lambda)
    \circ p'\bigr)}.
\label{eq:18}
\end{equation}
Thus, since $(p',p) \mapsto (-p',p \ominus (1 - 2\lambda) \circ p')$ is smooth
on $\overline{\mathbb{R}}^2$,
$\widetilde{f^*} \in \widetilde{\mathcal{F}}(\mathbb{R}^2)$ and hence
$f^* \in \mathcal{F}(\mathbb{R}^2)$. The conjugation $*$ has the following
properties.

\begin{theorem}
For $f,g \in \mathcal{F}(\mathbb{R}^2)$
\begin{enumerate}[(i)]
\item\label{item:3a} $f \mapsto f^*$ is anti-linear and continuous on
$\mathcal{F}(\mathbb{R}^2)$,
\item\label{item:3b} $(f \star g)^* = g^* \star f^*$,
\item\label{item:3c} $(f^*)^* = f$,
\item\label{item:3d} $\int_{\mathbb{R}^2} f^*(q,p) \ud{q}\ud{\mu(p)} =
\int_{\mathbb{R}^2} \bar{f}(q,p) \ud{q}\ud{\mu(p)}$,
\item\label{item:3e} $(\partial_q f)^* = \partial_q f^*$ and
$(D_p f)^* = D_p f^*$.
\end{enumerate}
\end{theorem}

\begin{proof}
(\ref{item:3a}) By virtue of \eqref{eq:10} and \eqref{eq:18}
\begin{equation}
\begin{aligned}
D_{p'}\widetilde{f^*}(p',p) & = -\overline{(D_{p'}\tilde{f})\bigl(-p',p \ominus
    (1 - 2\lambda) \circ p'\bigr)}
    - (1 - 2\lambda) \overline{(D_p\tilde{f})\bigl(-p',p \ominus (1 - 2\lambda)
    \circ p'\bigr)}, \\
D_p\widetilde{f^*}(p',p) & = \overline{(D_p\tilde{f})\bigl(-p',p \ominus
    (1 - 2\lambda) \circ p'\bigr)}.
\end{aligned}
\end{equation}
Hence
\begin{equation}
D_{p'}^n D_p^m\widetilde{f^*}(p',p) = (-1)^n \sum_{k=0}^n \binom{n}{k}
    (1 - 2\lambda)^k \overline{(D_{p'}^{n-k} D_p^{k+m} \tilde{f})
    \bigl(-p',p \ominus (1 - 2\lambda) \circ p'\bigr)}
\end{equation}
and
\begin{equation}
\norm{\widetilde{f^*}}_{n,m} \leq \sum_{k=0}^n \binom{n}{k} \abs{1 - 2\lambda}^k
    \norm{\tilde{f}}_{n-k,k+m},
\end{equation}
from which follows continuity of the conjugation $*$. Anti-linearity is an
immediate consequence of the definition.

(\ref{item:3b}) From \eqref{eq:16} we find that
\begin{equation}
\mathcal{F}_\beta (f^*)(q,p) = e^{\frac{i}{\hbar}(1 - 2\lambda)(q,p)}
    \mathcal{F}_\beta \bar{f}(q,p).
\label{eq:19}
\end{equation}
With the help of \eqref{eq:15}, \eqref{eq:20}, and \eqref{eq:19} we calculate
that
\begin{align}
& \mathcal{F}_\beta\bigl((f \star g)^*\bigr)(q,p) =
    e^{\frac{i}{\hbar}(1 - 2\lambda)(q,p)}
    \mathcal{F}_\beta\overline{(f \star g)}(q,p)
= e^{\frac{i}{\hbar}(1 - 2\lambda)(q,p)}
    \overline{\mathcal{F}_\beta(f \star g)(-q,-p)} \nonumber \\
& \quad = \frac{1}{2\pi\hbar} e^{\frac{i}{\hbar}(1 - 2\lambda)(q,p)}
    \overline{(\mathcal{F}_\beta f \diamond \mathcal{F}_\beta g)(-q,-p)}
    \nonumber \\
& \quad = \frac{1}{2\pi\hbar} \int_{\mathbb{R}^2}
    \mathcal{F}_\beta \bar{f}(-q',-p')
    \mathcal{F}_\beta \bar{g}(q + q',p \oplus p')
    e^{\frac{i}{\hbar}(1 - \lambda)(q + q',p')}
    e^{-\frac{i}{\hbar}\lambda(q',p \oplus p')}
    e^{\frac{i}{\hbar}(1 - 2\lambda)(q,p)} \ud{q'}\ud{\mu(p')} \nonumber \\
& \quad = \frac{1}{2\pi\hbar} \int_{\mathbb{R}^2}
    \mathcal{F}_\beta(f^*)(-q',-p')
    \mathcal{F}_\beta(g^*)(q + q',p \oplus p')
    e^{-\frac{i}{\hbar}(1 - \lambda)(q',p \oplus p')}
    e^{\frac{i}{\hbar}\lambda(q + q',p')} \ud{q'}\ud{\mu(p')} \nonumber \\
& \quad = \frac{1}{2\pi\hbar} \int_{\mathbb{R}^2}
    \mathcal{F}_\beta(f^*)(q - q',p \ominus p') \mathcal{F}_\beta(g^*)(q',p')
    e^{\frac{i}{\hbar}(1 - \lambda)(q - q',p')}
    e^{-\frac{i}{\hbar}\lambda(q',p \ominus p')}\ud{q'}\ud{\mu(p')} \nonumber \\
& \quad = \frac{1}{2\pi\hbar}
    \bigl(\mathcal{F}_\beta(g^*) \diamond \mathcal{F}_\beta(f^*)\bigr)(q,p)
= \mathcal{F}_\beta(g^* \star f^*)(q,p).
\end{align}

(\ref{item:3c}) Using \eqref{eq:19} we receive
\begin{align}
\mathcal{F}_\beta\bigl((f^*)^*\bigr)(q,p) & =
    e^{\frac{i}{\hbar}(1 - 2\lambda)(q,p)}\mathcal{F}_\beta\overline{(f^*)}(q,p)
= e^{\frac{i}{\hbar}(1 - 2\lambda)(q,p)}\overline{\mathcal{F}_\beta(f^*)(-q,-p)}
    \nonumber \\
& = e^{\frac{i}{\hbar}(1 - 2\lambda)(q,p)}
    \overline{e^{\frac{i}{\hbar}(1 - 2\lambda)(q,p)}
    \mathcal{F}_\beta \bar{f}(-q,-p)}
= \mathcal{F}_\beta f(q,p).
\end{align}

(\ref{item:3d}) Again using \eqref{eq:19} we find that
\begin{equation}
\int_{\mathbb{R}^2} f^*(q,p) \ud{q}\ud{\mu(p)} =
    2\pi\hbar \mathcal{F}_\beta(f^*)(0)
= 2\pi\hbar \mathcal{F}_\beta \bar{f}(0)
= \int_{\mathbb{R}^2} \bar{f}(q,p) \ud{q}\ud{\mu(p)}.
\end{equation}

(\ref{item:3e}) It follows directly from the property \eqref{eq:5b} of the
Fourier transform.
\end{proof}

From the above theorem it follows that the conjugation $*$ is an involution on
the algebra $\mathcal{F}(\mathbb{R}^2)$. As a consequence $*$ is an
anti-automorphism on $\mathcal{F}(\mathbb{R}^2)$.

\section{Extension of the $\star$-product to the Hilbert space $\mathcal{L}$}
\label{sec:5}
In what follows we will extend the \nbr{\star}product and involution $*$ to the
Hilbert space $\mathcal{L}$ introduced at the end of Section~\ref{sec:3}.
On the space $\mathcal{F}(\mathbb{R}^2)$ we define a scalar product by the
following formula
\begin{equation}
(f,g) = \int_{\mathbb{R}^2} f^* \star g \ud{l}.
\label{eq:22}
\end{equation}
From \eqref{eq:21} and \eqref{eq:19} we get that
\begin{equation}
(f,g) = \int_{\mathbb{R}^2} \overline{f(x)} g(x) \ud{l(x)},
\label{eq:25}
\end{equation}
thus $(\sdot,\sdot)$ is the usual scalar product from the Hilbert space
$\mathcal{L}$. The norm corresponding to the scalar product
$(\sdot,\sdot)$ will be denoted by $\norm{\sdot}_2$. Immediately from
\eqref{eq:22} we get that
\begin{equation}
\norm{f^*}_2 = \norm{f}_2,
\end{equation}
from which it follows that the conjugation $*$ is continuous with respect to
the \nbr{L^2}norm. Moreover, the following inequality holds
\begin{equation}
\norm{f \star g}_2 \leq \norm{f}_2 \norm{g}_2,
\label{eq:23}
\end{equation}
the consequent of which is continuity of the \nbr{\star}product in the
\nbr{L^2}norm. The proof of \eqref{eq:23} becomes straightforward if we make use
of an operator representation introduced in Section~\ref{sec:8}.

Since $\mathcal{F}(\mathbb{R}^2)$ is dense in $\mathcal{L}$ and
the conjugation $*$ and the \nbr{\star}product are continuous in the
\nbr{L^2}norm, these operations can be uniquely extended to the whole space
$\mathcal{L}$. Note, that $\mathcal{L}$ is a Banach algebra as well as a
Hilbert algebra.

We will also introduce a trace functional
\begin{equation}
\tr(f) = \int_{\mathbb{R}^2} f(x) \ud{l(x)}
\end{equation}
for $f \in L^1(\mathbb{R}^2,\dd{l})$. The scalar product on $\mathcal{L}$
takes then the following form
\begin{equation}
(f,g) = \tr(f^* \star g).
\end{equation}

\section{Extension of the $\star$-product to an algebra of distributions}
\label{sec:6}
In what follows we will extend the \nbr{\star}product to a suitable space of
distributions. The algebra $\mathcal{F}(\mathbb{R}^2)$ will play the role of
the space of test functions. We will denote by $\mathcal{F}'(\mathbb{R}^2)$ the
space of continuous linear functionals on $\mathcal{F}(\mathbb{R}^2)$, i.e.
distributions. The dual space $\mathcal{F}'(\mathbb{R}^2)$ is endowed with the
strong dual topology, that of uniform convergence on bounded subsets of
$\mathcal{F}(\mathbb{R}^2)$. For $f \in \mathcal{F}'(\mathbb{R}^2)$ we will
denote by $\braket{f,h}$ the value of the functional $f$ at
$h \in \mathcal{F}(\mathbb{R}^2)$. We will identify functions
$f \in \mathcal{F}(\mathbb{R}^2)$ with the following functionals
\begin{equation}
h \mapsto \int_{\mathbb{R}^2} f \star h \ud{q}\ud{\mu(p)}.
\label{eq:24}
\end{equation}
These functionals are continuous on $\mathcal{F}(\mathbb{R}^2)$. Indeed, the
integral is continuous on $\mathcal{F}(\mathbb{R}^2)$ since
\begin{equation}
\Abs{\int_{\mathbb{R}^2} f(q,p) \ud{q}\ud{\mu(p)}} =
    \Abs{\int_\mathbb{R} \tilde{f}(0,p) \ud{\mu(p)}}
\leq \int_\mathbb{R} \abs{\tilde{f}(0,p)} \ud{\mu(p)}
\leq \frac{\pi}{\sqrt\beta} \norm{\tilde{f}}_{0,0},
\end{equation}
which together with the continuity of the \nbr{\star}product implies the
continuity of the functionals \eqref{eq:24}. Thus we may write
$\mathcal{F}(\mathbb{R}^2) \subset \mathcal{F}'(\mathbb{R}^2)$. The functionals
\eqref{eq:24}, by virtue of \eqref{eq:25}, can be also written in the form
\begin{equation}
h \mapsto \int_{\mathbb{R}^2} f(q,p) Sh(q,p) \ud{q}\ud{\mu(p)},
\label{eq:26}
\end{equation}
where
\begin{equation}
Sh(q,p) = \overline{h^*}(q,p)
= \frac{1}{2\pi\hbar} \int_{\mathbb{R}^2} \mathcal{F}_\beta h(q',p')
    e^{-\frac{i}{\hbar}(1 - 2\lambda)(q',p')}
    e^{\frac{i}{\hbar}(q,p')} e^{-\frac{i}{\hbar}(q',p)} \ud{q'}\ud{\mu(p')}
\end{equation}
is a topological isomorphism of the vector space $\mathcal{F}(\mathbb{R}^2)$.
Note, that $S$ commutes with partial derivatives $\partial_q$, $D_p$ and that
it reduces to the identity operator under the integral sign, i.e.
\begin{equation}
\int_{\mathbb{R}^2} Sf(q,p) \ud{q}\ud{\mu(p)} =
    \int_{\mathbb{R}^2} f(q,p) \ud{q}\ud{\mu(p)}.
\end{equation}
Moreover, for $\lambda = \tfrac{1}{2}$ the operator $S$ is equal to the identity
operator. If we explicitly denote dependence of the \nbr{\star}product on
$\lambda$ by writing $\star_\lambda$, then we get the following property
of the operator $S$
\begin{equation}
S(f \star_\lambda g) = Sf \star_{1-\lambda} Sg.
\end{equation}

Formula \eqref{eq:26} allows for identification of a broad class of functions
with distributions. In particular, function identically equal to 1 can be
identified with the following distribution
\begin{equation}
h \mapsto \int_{\mathbb{R}^2} h(q,p) \ud{q}\ud{\mu(p)}.
\end{equation}

The Fourier transform, partial differentiation, and ordinary multiplication by
a function can be extended to the space of distributions
$\mathcal{F}'(\mathbb{R}^2)$ in the following way
\begin{subequations}
\begin{gather}
\braket{\partial_q f,h} = -\braket{f,\partial_q h}, \quad
\braket{D_p f,h} = -\braket{f,D_p h}, \\
\braket{\mathcal{F}_\beta f,h} = \braket{f,S^{-1}\mathcal{F}_\beta S\check{h}},
    \\
\braket{f \cdot g,h} = \braket{f,S^{-1}(gSh)},
\end{gather}
\end{subequations}
where $f \in \mathcal{F}'(\mathbb{R}^2)$, $h \in \mathcal{F}(\mathbb{R}^2)$,
$\check{h}(q,p) = h(-q,-p)$, and $g$ is any function on $\mathbb{R}^2$ such that
$gh \in \mathcal{F}(\mathbb{R}^2)$ for every $h \in \mathcal{F}(\mathbb{R}^2)$.

For $f \in \mathcal{F}'(\mathbb{R}^2)$ and $g \in \mathcal{F}(\mathbb{R}^2)$
we define $f \star g \in \mathcal{F}'(\mathbb{R}^2)$ and $g \star f \in
\mathcal{F}'(\mathbb{R}^2)$ by the formulas
\begin{equation}
\braket{f \star g,h} = \braket{f,g \star h}, \quad
\braket{g \star f,h} = \braket{f,h \star g} \quad
\text{for every $h \in \mathcal{F}(\mathbb{R}^2)$}.
\label{eq:27}
\end{equation}
Note, that the maps $g \mapsto f \star g$ and $g \mapsto g \star f$ are
continuous from $\mathcal{F}(\mathbb{R}^2)$ to $\mathcal{F}'(\mathbb{R}^2)$,
since the maps $g \mapsto \braket{f,g \star h}$ and
$g \mapsto \braket{f,h \star g}$ are continuous, uniformly for $h$ in a bounded
subset of $\mathcal{F}(\mathbb{R}^2)$.

Denote by $\mathcal{F}_\star(\mathbb{R}^2)$ the following subspace of
distributions:
\begin{equation}
\mathcal{F}_\star(\mathbb{R}^2) = \{f \in \mathcal{F}'(\mathbb{R}^2) \mid
    \text{$f \star g$ and $g \star f \in \mathcal{F}(\mathbb{R}^2)$ for every
    $g \in \mathcal{F}(\mathbb{R}^2)$}\}.
\end{equation}
In particular, $\mathcal{F}(\mathbb{R}^2) \subset
\mathcal{F}_\star(\mathbb{R}^2)$. For $f \in \mathcal{F}_\star(\mathbb{R}^2)$
the maps $g \mapsto f \star g$ and $g \mapsto g \star f$ are continuous from
$\mathcal{F}(\mathbb{R}^2)$ to $\mathcal{F}(\mathbb{R}^2)$ by the closed graph
theorem. Thus, for $f,g \in \mathcal{F}_\star(\mathbb{R}^2)$ we may define their
\nbr{\star}product by the formula
\begin{equation}
\braket{f \star g,h} = \braket{f,g \star h} = \braket{g,h \star f} \quad
\text{for every $h \in \mathcal{F}(\mathbb{R}^2)$}.
\label{eq:28}
\end{equation}
The second equality in the above definition is indeed satisfied for every
$h \in \mathcal{F}(\mathbb{R}^2)$, which can be easily proved for
$h = h_1 \star h_2$ ($h_1,h_2 \in \mathcal{F}(\mathbb{R}^2)$) and the general
case follows from the fact that $\mathcal{F}(\mathbb{R}^2) \star
\mathcal{F}(\mathbb{R}^2)$ is linearly dense in $\mathcal{F}(\mathbb{R}^2)$.
Straightforward calculations with the use of \eqref{eq:27} and \eqref{eq:28}
verify that $f \star g \in \mathcal{F}_\star(\mathbb{R}^2)$ and the
associativity of the \nbr{\star}product. Note, that
$1 \in \mathcal{F}_\star(\mathbb{R}^2)$ and $f \star 1 = 1 \star f = f$ for
every $f \in \mathcal{F}_\star(\mathbb{R}^2)$.

The involution $*$ can be extended to the algebra
$\mathcal{F}_\star(\mathbb{R}^2)$ in a natural way:
\begin{equation}
\braket{f^*,h} = \overline{\braket{f,h^*}} \quad
\text{for every $h \in \mathcal{F}(\mathbb{R}^2)$}
\end{equation}
and $f \in \mathcal{F}_\star(\mathbb{R}^2)$. In particular, $1^* = 1$. Thus,
$\mathcal{F}_\star(\mathbb{R}^2)$ is an involutive algebra with unity, being
a natural extension of the algebra $\mathcal{F}(\mathbb{R}^2)$.

\begin{theorem}
If $\phi(p)$ is a smooth function on $\overline{\mathbb{R}}$, then
$\phi \in \mathcal{F}_\star(\mathbb{R}^2)$.
\end{theorem}

\begin{proof}
For $g,h \in \mathcal{F}(\mathbb{R}^2)$ we have
\begin{align}
& \braket{\phi \star g,h} = \braket{\phi,g \star h}
= \int_{\mathbb{R}^2} \phi(p) S(g \star h)(q,p) \ud{q}\ud{\mu(p)}
= \int_{\mathbb{R}^2} \phi(p) \overline{(h^* \star g^*)(q,p)} \ud{q}\ud{\mu(p)}
    \nonumber \\
& \quad = \int_\mathbb{R} \phi(p)
    \overline{\bigl(\widetilde{h^*} \odot \widetilde{g^*}\bigr)(0,p)}\ud{\mu(p)}
= \frac{1}{2\pi\hbar} \int_{\mathbb{R}^2} \phi(p)
    \overline{\widetilde{h^*}\bigl(p',p \ominus \lambda \circ p'\bigr)
    \widetilde{g^*}\bigl(-p',p \ominus (1 - \lambda) \circ p'\bigr)}
    \ud{\mu(p)}\ud{\mu(p')} \nonumber \\
& \quad = \frac{1}{2\pi\hbar} \int_{\mathbb{R}^2} \phi(p)
    \widetilde{Sh}\bigl(-p',p \ominus \lambda \circ p'\bigr)
    \widetilde{Sg}\bigl(p',p \ominus (1 - \lambda) \circ p'\bigr)
    \ud{\mu(p)}\ud{\mu(p')} \nonumber \\
& \quad = \frac{1}{2\pi\hbar} \int_{\mathbb{R}^2}\phi(p \oplus \lambda \circ p')
    \widetilde{Sg}\bigl(p',p \ominus (1 - 2\lambda) \circ p'\bigr)
    \widetilde{Sh}(-p',p) \ud{\mu(p)}\ud{\mu(p')} \nonumber \\
& \quad = \frac{1}{2\pi\hbar} \int_{\mathbb{R}^3}\phi(p \oplus \lambda \circ p')
    \widetilde{Sg}\bigl(p',p \ominus (1 - 2\lambda) \circ p'\bigr)
    e^{\frac{i}{\hbar}(q,p')} Sh(q,p) \ud{q}\ud{\mu(p)}\ud{\mu(p')} \nonumber \\
& \quad = \int_{\mathbb{R}^2} \left( \frac{1}{2\pi\hbar} \int_\mathbb{R}
    \phi(p \oplus \lambda \circ p') \tilde{g}(p',p) e^{\frac{i}{\hbar}(q,p')}
    \ud{\mu(p')} \right) Sh(q,p) \ud{q}\ud{\mu(p)}.
\end{align}
Hence
\begin{equation}
(\phi \star g)(q,p) = \frac{1}{2\pi\hbar} \int_\mathbb{R}
    \phi(p \oplus \lambda \circ p') \tilde{g}(p',p) e^{\frac{i}{\hbar}(q,p')}
    \ud{\mu(p')}.
\label{eq:37}
\end{equation}
Since $\phi \star g$ is an inverse generalized Fourier transform in position
variable of a smooth function on $\overline{\mathbb{R}}^2$,
$\phi \star g \in \mathcal{F}(\mathbb{R}^2)$. Similarly we can prove that
$g \star \phi \in \mathcal{F}(\mathbb{R}^2)$. Therefore
$\phi \in \mathcal{F}_\star(\mathbb{R}^2)$.
\end{proof}

Note, that $q$, $p$ and their natural powers are not elements of
$\mathcal{F}'(\mathbb{R}^2)$, since the integral in \eqref{eq:26} will not be
well defined for every $h \in \mathcal{F}(\mathbb{R}^2)$. However, the integral
formula for the \nbr{\star}product can still make sense for these and other
functions which are not in $\mathcal{F}'(\mathbb{R}^2)$. In particular, for
$f(q,p) = q^n \phi(p)$ where $n \in \mathbb{N}$ and $\phi$ is smooth on
$\overline{\mathbb{R}}$ we can define
\begin{equation}
(f \star g)^\sim(p',p) = \frac{1}{2\pi\hbar} \int_\mathbb{R} q^n \left(
    \int_\mathbb{R} \phi\bigl(p \oplus \lambda \circ (p' \ominus p'')\bigr)
    \tilde{g}\bigl(p' \ominus p'',p \ominus (1 - \lambda) \circ p''\bigr)
    e^{-\frac{i}{\hbar}(q,p'')} \ud{\mu(p'')} \right) \ud{q}
\label{eq:41}
\end{equation}
for all $g \in \mathcal{F}(\mathbb{R}^2)$ for which the integrals are
convergent. Note, that the integral with respect to $q$, in general, will have
to be an improper integral \eqref{eq:40}.

\section{$C^*$-algebra of observables and states}
\label{sec:7}
On the space $\mathcal{F}(\mathbb{R}^2)$ we can introduce another norm according
to the formula
\begin{equation}
\norm{f} = \sup\{\norm{f \star g}_2 \mid g \in \mathcal{F}(\mathbb{R}^2),\ 
    \norm{g}_2 = 1\}.
\end{equation}
This is a \nbr{C^*}norm, i.e. it satisfies
\begin{enumerate}[(i)]
\item $\norm{f \star g} \leq \norm{f}\norm{g}$,
\item $\norm{f^*} = \norm{f}$,
\item $\norm{f^* \star f} = \norm{f}^2$,
\end{enumerate}
for $f,g \in \mathcal{F}(\mathbb{R}^2)$. Indeed, this follows directly from the
fact that $f \star {}$ is a bounded linear operator on the Hilbert space
$\mathcal{L}$ defined on a dense domain $\mathcal{F}(\mathbb{R}^2)$
(the boundedness of $f \star {}$ can be seen from \eqref{eq:23}). The norm
$\norm{f}$ is then defined as the operator norm of the operator $f \star {}$.
Since $\norm{f}$ is the smallest constant $C$ satisfying the inequality
\begin{equation}
\norm{f \star g}_2 \leq C\norm{g}_2
\text{ for all $g \in \mathcal{F}(\mathbb{R}^2)$},
\end{equation}
it is clear from \eqref{eq:23} that $\norm{f} \leq \norm{f}_2$, and so
convergence in \nbr{L^2}norm implies convergence in the norm $\norm{\sdot}$.

The space $\mathcal{F}(\mathbb{R}^2)$ is not complete with respect to the
\nbr{C^*}norm $\norm{\sdot}$, thus $\mathcal{F}(\mathbb{R}^2)$ is only a
pre-$C^*$-algebra. The completion of $\mathcal{F}(\mathbb{R}^2)$ to a
\nbr{C^*}algebra will be denoted by $\mathcal{A}(\mathbb{R}^2)$. The algebra
$\mathcal{A}(\mathbb{R}^2)$ is a \nbr{C^*}algebra of observables. With its help
we can define states, in a standard way, as continuous positive linear
functionals on $\mathcal{A}(\mathbb{R}^2)$ normalized to unity, i.e. a
continuous linear functional $\Lambda \colon \mathcal{A}(\mathbb{R}^2) \to
\mathbb{C}$ is a state if
\begin{enumerate}[(i)]
\item $\norm{\Lambda} = 1$,
\item $\Lambda(f^* \star f) \geq 0$ for every $f \in \mathcal{A}(\mathbb{R}^2)$.
\end{enumerate}
The set of all states is convex. Pure states are defined as extreme points of
this set, i.e. as those states which cannot be written as convex linear
combinations of some other states. In other words $\Lambda_\text{pure}$ is a
pure state if and only if there do not exist two different states $\Lambda_1$
and $\Lambda_2$ such that $\Lambda_\text{pure} = p\Lambda_1 + (1 - p)\Lambda_2$
for some $p \in (0,1)$.

The expectation value of an observable $f \in \mathcal{A}(\mathbb{R}^2)$ in
a state $\Lambda$ is from definition equal
\begin{equation}
\braket{f}_\Lambda = \Lambda(f).
\end{equation}
If $f$ is self-adjoint, i.e. $f^* = f$, then $\braket{f}_\Lambda \in\mathbb{R}$.

The next two theorems provide characterization of states in terms of
quasi-probabilistic distribution functions $\rho \in \mathcal{L}$.
They follow directly from the operator representation introduced below.

\begin{theorem}
If $\rho \in \mathcal{L}$ satisfies
\begin{enumerate}[(i)]
\item\label{item:4a} $\rho^* = \rho$,
\item\label{item:4b} $\int_{\mathbb{R}^2} \rho \ud{l} = 1$,
\item\label{item:4c} $\int_{\mathbb{R}^2} f^* \star f \star \rho \ud{l} \geq 0$
for every $f \in \mathcal{F}(\mathbb{R}^2)$,
\end{enumerate}
then the functional
\begin{equation}
\Lambda_\rho(f) = \int_{\mathbb{R}^2} f \star \rho \ud{l}
\label{eq:29}
\end{equation}
is a state. Vice verse, every state $\Lambda$ can be written in the form
\eqref{eq:29} for some $\rho \in \mathcal{L}$ and satisfying properties
(\ref{item:4a})--(\ref{item:4c}). The representation \eqref{eq:29} is unique.
\end{theorem}

\begin{theorem}
A state $\Lambda_\rho$ is pure if and only if the corresponding function $\rho$
is idempotent, i.e.
\begin{equation}
\rho \star \rho = \rho.
\end{equation}
\end{theorem}

From \eqref{eq:29} the expectation value of an observable
$f \in \mathcal{F}(\mathbb{R}^2)$ in a state $\Lambda_\rho$ can be written in
a form
\begin{equation}
\braket{f}_\rho = \int_{\mathbb{R}^2} f \star \rho \ud{l}.
\end{equation}
The above formula can be extended in a direct way to general observables
$f \in \mathcal{F}_\star(\mathbb{R}^2)$ provided that $f \star \rho$ is well
defined. If $\rho \in \mathcal{F}(\mathbb{R}^2)$ then the product $f \star \rho$
is a well defined function in $\mathcal{F}(\mathbb{R}^2)$ and we can treat
$f \star {}$ as a densely defined operator on $\mathcal{L}$,
which sometimes might be extended to a wider subspace of functions.

\section{Operator representation}
\label{sec:8}
By virtue of Gelfand-Naimark theorem the \nbr{C^*}algebra of observables
$\mathcal{A}(\mathbb{R}^2)$ can be represented as an algebra of bounded
linear operators on a certain Hilbert space $\mathcal{H}$. In what follows we
will present an explicit construction of this representation for
$\mathcal{H} = L^2(\mathbb{R},\dd{\mu})$. The Hilbert space $\mathcal{H}$ will
play the role of the space of states from a standard description of quantum
mechanics. The constructed representation will result, in fact, in a momentum
representation of a quantum system.

Let $\mathcal{H} = L^2(\mathbb{R},\dd{\mu})$ be a Hilbert space of square
integrable functions on $\mathbb{R}$ with a scalar product given by
\begin{equation}
(\varphi,\psi) = \int_\mathbb{R} \overline{\varphi(p)} \psi(p) \ud{\mu(p)},
\end{equation}
where $\dd{\mu}$ is a measure given by \eqref{eq:30}. For a function
$f \in \mathcal{F}(\mathbb{R}^2)$ we define an operator $\hat{f}$ acting in
$\mathcal{H}$ by the formula
\begin{equation}
\hat{f}\psi(p) = \frac{1}{2\pi\hbar} \int_\mathbb{R}
    \tilde{f}(p',p \ominus \lambda \circ p') \psi(p \ominus p') \ud{\mu(p')},
\label{eq:33}
\end{equation}
where $\lambda$ is the same as in the definition of the \nbr{\star}product.

Let us denote by $\mathsf{f}$ the integral kernel of the operator $\hat{f}$,
i.e.
\begin{equation}
\hat{f}\psi(p) = \int_\mathbb{R} \mathsf{f}(p,p') \psi(p') \ud{\mu(p')}.
\end{equation}
The function $f(q,p)$ can be expressed by the integral kernel of the
corresponding operator $\hat{f}$ in the following way
\begin{equation}
f(q,p) = \int_\mathbb{R} \mathsf{f}\bigl(p \oplus \lambda \circ p',
    p \ominus (1 - \lambda) \circ p'\bigr)
    e^{\frac{i}{\hbar}(q,p')} \ud{\mu(p')}.
\end{equation}
Indeed, from the above equation the generalized Fourier transform in position
variable of the function $f$ is expressed by the integral kernel $\mathsf{f}$
in the following way
\begin{equation}
\tilde{f}(p',p) = 2\pi\hbar
    \mathsf{f}\bigl(p \oplus \lambda \circ p',
    p \ominus (1 - \lambda) \circ p'\bigr).
\end{equation}
Then
\begin{align}
\hat{f}\psi(p) & = \frac{1}{2\pi\hbar} \int_\mathbb{R}
    \tilde{f}(p',p \ominus \lambda \circ p') \psi(p \ominus p')
    \ud{\mu(p')}
= \int_\mathbb{R} \mathsf{f}(p,p \ominus p') \psi(p \ominus p')
    \ud{\mu(p')} \nonumber \\
& = \int_\mathbb{R} \mathsf{f}(p,p') \psi(p') \ud{\mu(p')}.
\end{align}

Performing the following change of variables
\begin{equation}
\begin{aligned}
a & = p \oplus \lambda \circ p', & p' & = a \ominus b, \\
b & = p \ominus (1 - \lambda) \circ p', &
p & = (1 - \lambda) \circ a \oplus \lambda \circ b,
\end{aligned}
\end{equation}
the integral kernel $\mathsf{f}$ is expressed by the function $f$ in the
following way
\begin{equation}
\mathsf{f}(a,b) = \frac{1}{2\pi\hbar} \tilde{f}\bigl(a \ominus b,
    (1 - \lambda) \circ a \oplus \lambda \circ b\bigr).
\label{eq:31}
\end{equation}
Note, that the transformation
$(a,b) \mapsto (a \ominus b, (1 - \lambda) \circ a \oplus \lambda \circ b)$ is
a smooth bijection of $\overline{\mathbb{R}}^2$ onto $\overline{\mathbb{R}}^2$,
which inverse is also smooth. Thus, $f \mapsto \hat{f}$ is a linear isomorphism
of $\mathcal{F}(\mathbb{R}^2)$ onto the space of integral operators whose
integral kernels $\mathsf{f} \in \widetilde{\mathcal{F}}(\mathbb{R}^2)$.

\begin{theorem}
\label{thm:8}
For $f,g \in \mathcal{F}(\mathbb{R}^2)$
\begin{enumerate}[(i)]
\item\label{item:5a} $\widehat{f \star g} = \hat{f} \hat{g}$,
\item\label{item:5b} $\widehat{f^*} = \hat{f}^\dagger$,
\item\label{item:5c} $\hat{f}$ is a trace class operator and
$\Tr(\hat{f}) = \tr(f)$,
\item\label{item:5d} the Hilbert-Schmidt scalar product of operators $\hat{f}$
and $\hat{g}$ is equal $(\hat{f},\hat{g}) \equiv \Tr(\hat{f}^\dagger \hat{g})
= (f,g)$,
\item\label{item:5e} $\norm{\hat{f}} = \norm{f}$.
\end{enumerate}
\end{theorem}

\begin{proof}
(\ref{item:5a}) We have that
\begin{equation}
\hat{f}\hat{g}\psi(p) = \frac{1}{(2\pi\hbar)^2} \int_{\mathbb{R}^2}
    \tilde{f}(p'',p \ominus \lambda \circ p'')
    \tilde{g}(p',p \ominus p'' \ominus \lambda \circ p')
    \psi(p \ominus p'' \ominus p') \ud{\mu(p')}\ud{\mu(p'')}.
\end{equation}
Using the translational invariance of the measure $\mu$ when integrating over
$p'$ we can replace $p'$ by $p' \ominus p''$ under the integral sign receiving
\begin{align}
\hat{f}\hat{g}\psi(p) & = \frac{1}{(2\pi\hbar)^2} \int_{\mathbb{R}^2}
    \tilde{f}\bigl(p'',p \ominus \lambda \circ p''\bigr)
    \tilde{g}\bigl(p' \ominus p'',p \ominus p'' \ominus \lambda \circ
    (p' \ominus p'')\bigr) \psi\bigl(p \ominus p'\bigr)
    \ud{\mu(p')}\ud{\mu(p'')} \nonumber \\
& = \frac{1}{2\pi\hbar} \int_\mathbb{R} \Biggl(
    \frac{1}{2\pi\hbar} \int_\mathbb{R}
    \tilde{f}\bigl(p'',p \ominus \lambda \circ p' \oplus
    \lambda \circ (p' \ominus p'')\bigr) \nonumber \\
& \quad \times \tilde{g}\bigl(p' \ominus p'',p \ominus \lambda \circ
    p' \ominus (1 - \lambda) \circ p''\bigr) \ud{\mu(p'')} \Biggr)
    \psi\bigl(p \ominus p'\bigr) \ud{\mu(p')} \nonumber \\
& = \frac{1}{2\pi\hbar} \int_\mathbb{R}
    (\tilde{f} \odot \tilde{g})(p',p \ominus \lambda \circ p')
    \psi(p \ominus p') \ud{\mu(p')}
= (\widehat{f \star g})\psi(p).
\end{align}

(\ref{item:5b}) From \eqref{eq:18} and translational invariance of the measure
$\mu$ we get
\begin{align}
(\varphi,\widehat{f^*}\psi) & = \frac{1}{2\pi\hbar} \int_{\mathbb{R}^2}
    \overline{\varphi(p)} \widetilde{f^*}(p',p \ominus \lambda \circ p')
    \psi(p \ominus p') \ud{\mu(p')}\ud{\mu(p)} \nonumber \\
& = \frac{1}{2\pi\hbar} \int_{\mathbb{R}^2}
    \overline{\tilde{f}(-p',p \ominus (1 - \lambda) \circ p') \varphi(p)}
    \psi(p \ominus p') \ud{\mu(p')}\ud{\mu(p)} \nonumber \\
& = \frac{1}{2\pi\hbar} \int_{\mathbb{R}^2}
    \overline{\tilde{f}(-p',p \oplus \lambda \circ p') \varphi(p \oplus p')}
    \psi(p) \ud{\mu(p')}\ud{\mu(p)} \nonumber \\
& = (\hat{f}\varphi,\psi) = (\varphi,\hat{f}^\dagger \psi).
\end{align}

(\ref{item:5c}) Since the Hilbert space $L^2(\mathbb{R},\dd{\mu})$ is unitarily
equivalent with a Hilbert space of square integrable functions on a circle and
the integral kernel $\mathsf{f}$ of the operator $\hat{f}$ is smooth on
$\overline{\mathbb{R}}^2$, then by
\cite[Proposition~3.5, page~174]{Sugiura:1990} the operator $\hat{f}$ is of
trace class with the trace given by the formula
\begin{equation}
\Tr(\hat{f}) = \int_\mathbb{R} \mathsf{f}(p,p) \ud{\mu(p)}.
\end{equation}
Using \eqref{eq:31} we calculate that
\begin{equation}
\Tr(\hat{f}) = \int_\mathbb{R} \mathsf{f}(p,p) \ud{\mu(p)}
= \frac{1}{2\pi\hbar} \int_\mathbb{R} \tilde{f}(0,p) \ud{\mu(p)}
= \frac{1}{2\pi\hbar} \int_{\mathbb{R}^2} f(q,p) \ud{q}\ud{\mu(p)}
= \tr(f).
\end{equation}

(\ref{item:5d}) This property immediately follows from
(\ref{item:5a})--(\ref{item:5c}).

(\ref{item:5e}) The operator norm of a bounded operator $\hat{f}$ can be
expressed in terms of the Hilbert-Schmidt norm $\norm{\sdot}_2$ according to the
formula
\begin{equation}
\norm{\hat{f}} = \sup\{\norm{\hat{f}\hat{g}}_2 \mid
\text{$\hat{g}$ is a Hilbert-Schmidt operator and $\norm{\hat{g}}_2 = 1$}\}.
\end{equation}
Using properties (\ref{item:5a}) and (\ref{item:5d}) and the fact that
$\mathcal{F}(\mathbb{R}^2)$ is dense in $\mathcal{L}$ we get the result.
\end{proof}

The above theorem states that the map $f \mapsto \hat{f}$ is a faithful
\nbr{*}representation of the algebra $\mathcal{F}(\mathbb{R}^2)$ on the Hilbert
space $\mathcal{H}$. From property (\ref{item:5d}) this representation can be
extended to the algebra $\mathcal{L}$ resulting in a Hilbert algebra isomorphism
of $\mathcal{L}$ onto the space of Hilbert-Schmidt operators
$\mathcal{B}_2(\mathcal{H})$. Moreover, from property (\ref{item:5e}) we can
further extend this representation to a representation of the \nbr{C^*}algebra
$\mathcal{A}(\mathbb{R}^2)$, which will give us a \nbr{C^*}algebra isomorphism
of $\mathcal{A}(\mathbb{R}^2)$ onto the \nbr{C^*}algebra of compact operators
$\mathcal{K}(\mathcal{H})$.

Let us consider the operator $\hat{f} = \psi(\varphi,\sdot)$ for
$\varphi,\psi \in \mathcal{H}$. Its integral kernel is equal
$\mathsf{f}(a,b) = \psi(a)\overline{\varphi(b)}$ and, therefore, the
corresponding function on phase space $f(q,p) = W_\lambda(\varphi,\psi)(q,p)$ is
equal
\begin{equation}
W_\lambda(\varphi,\psi)(q,p) = \int_\mathbb{R}
    \overline{\varphi\bigl(p \ominus (1 - \lambda) \circ p'\bigr)}
    \psi\bigl(p \oplus \lambda \circ p'\bigr)
    e^{\frac{i}{\hbar}(q,p')} \ud{\mu(p')}.
\end{equation}
The functions $W_\lambda(\varphi,\psi)$ are the generalized \nbr{\lambda}Wigner
functions and are elements of the Hilbert space $\mathcal{L}$. The following
properties of the functions $W_\lambda(\varphi,\psi)$ are an immediate
consequence of Theorem~\ref{thm:8}.

\begin{theorem}
\label{thm:9}
For $\varphi,\psi,\phi,\chi \in \mathcal{H}$ and
$f \in \mathcal{F}(\mathbb{R}^2)$
\begin{enumerate}[(i)]
\item\label{item:6a} $W_\lambda(\varphi,\psi)^* = W_\lambda(\psi,\varphi)$,
\item\label{item:6b} $\int_{\mathbb{R}^2} W_\lambda(\varphi,\psi) \ud{l} =
(\varphi,\psi)$,
\item\label{item:6c} $\bigl(W_\lambda(\varphi,\psi),W_\lambda(\phi,\chi)\bigr)
= \overline{(\varphi,\phi)}(\psi,\chi)$,
\item\label{item:6d} $W_\lambda(\varphi,\psi) \star W_\lambda(\phi,\chi)
= (\varphi,\chi) W_\lambda(\phi,\psi)$,
\item\label{item:6e} $f \star W_\lambda(\varphi,\psi) =
W_\lambda(\varphi,\hat{f}\psi)$ and $W_\lambda(\varphi,\psi) \star f =
W_\lambda(\hat{f}^\dagger\varphi,\psi)$.
\end{enumerate}
\end{theorem}

We can define a tensor product of the Hilbert space $\mathcal{H}$ and its dual
$\mathcal{H}^*$ by the formula
\begin{equation}
\varphi^* \otimes_\lambda \psi = W_\lambda(\varphi,\psi),
\end{equation}
where $\varphi \mapsto \varphi^*$ is an anti-linear isomorphism of $\mathcal{H}$
onto $\mathcal{H}^*$ appearing in the Riesz representation theorem. The map
$\otimes_\lambda \colon \mathcal{H}^* \times \mathcal{H} \to \mathcal{L}$ is
clearly bilinear and from property (\ref{item:6c}) from Theorem~\ref{thm:9}
it satisfies
\begin{equation}
(\varphi^* \otimes_\lambda \psi,\phi^* \otimes_\lambda \chi) =
    (\varphi^*,\phi^*)(\psi,\chi).
\end{equation}
Moreover, since the set of generalized \nbr{\lambda}Wigner functions
$W_\lambda(\varphi,\psi)$ is linearly dense in $\mathcal{L}$ the map
$\otimes_\lambda$ indeed defines a tensor product of $\mathcal{H}^*$ and
$\mathcal{H}$ equal to $\mathcal{L}$.

For $f \in \mathcal{F}(\mathbb{R}^2)$ we can treat $f \star {}$ as an operator
on the Hilbert space $\mathcal{L}$. Then, by property (\ref{item:6e}) from
Theorem~\ref{thm:9} we get that
\begin{equation}
f \star {} = \hat{1} \otimes_\lambda \hat{f}.
\end{equation}

Note, that the operator representation of the algebra $\mathcal{L}$ gives a one
to one correspondence between states $\rho \in \mathcal{L}$ and density
operators $\hat{\rho}$, i.e. trace class operators satisfying
\begin{enumerate}[(i)]
\item $\hat{\rho}^\dagger = \hat{\rho}$,
\item $\Tr(\hat{\rho}) = 1$,
\item $(\varphi,\hat{\rho}\varphi) \geq 0$ for every $\varphi \in \mathcal{H}$.
\end{enumerate}
From this correspondence we can see that a function $\rho \in \mathcal{L}$ is
a state if and only if it is in the form
\begin{equation}
\rho = \sum_k p_k W_\lambda(\varphi_k,\varphi_k),
\end{equation}
where $\varphi_k \in \mathcal{H}$, $\norm{\varphi_k} = 1$, $p_k \geq 0$, and
$\sum_k p_k = 1$. In particular, $\rho \in \mathcal{L}$ is a pure state if and
only if
\begin{equation}
\rho = W_\lambda(\varphi,\varphi)
\label{eq:34}
\end{equation}
for some normalized vector $\varphi \in \mathcal{H}$.

The generalized \nbr{\lambda}Wigner functions obey the following probability
interpretation
\begin{equation}
\frac{1}{2\pi\hbar} \int_\mathbb{R} W_\lambda(\varphi,\varphi)(q,p) \ud{q} =
    \abs{\varphi(p)}^2.
\end{equation}

The operator representation can be extended to the algebra
$\mathcal{F}_\star(\mathbb{R}^2)$ in the following manner. Let
$\mathcal{D}$ denote a set of all functions on $\mathbb{R}$ which extend to
smooth functions on $\overline{\mathbb{R}}$. Then $\mathcal{D}$ is a dense
subspace of $\mathcal{H}$. If $\varphi,\psi \in \mathcal{D}$, then
$W_\lambda(\varphi,\psi) \in \mathcal{F}(\mathbb{R}^2)$. Hence, for
$f \in \mathcal{F}_\star(\mathbb{R}^2)$ we can define the operator $\hat{f}$ by
the following bilinear form
\begin{equation}
(\varphi,\hat{f}\psi) = \braket{f,W_\lambda(\varphi,\psi)}, \quad
\varphi,\psi \in \mathcal{D},
\label{eq:32}
\end{equation}
provided that this bilinear form is continuous with respect to the first
variable. The formula \eqref{eq:32} defines a unique possibly unbounded operator
on $\mathcal{H}$ with a dense domain $\mathcal{D}$. In the particular case where
$f \in \mathcal{F}(\mathbb{R}^2)$, this definition of the operator $\hat{f}$
coincides with the definition \eqref{eq:33}.

Let $\mathcal{S}$ denote a space of smooth functions on $\mathbb{R}$ which,
together with all derivatives, decay faster than any power of $p$. This space is
dense in $\mathcal{H}$. Let $\hat{q} = i\hbar D_p$ and $\hat{p} = p$ be
operators on $\mathcal{H}$ defined on the domain $\mathcal{S}$. They are
symmetric and satisfy the commutation relation \eqref{eq:2}. Moreover, as was
discussed in \cite{Kempf:1995}, only $\hat{p}$ is essentially self-adjoint. In
what follows we will show that these operators correspond to observables of
position and momentum $q$, $p$, and that an operator $\hat{f}$ corresponding to
a function $f$ is, in fact, a \nbr{\lambda}ordered function $f$ of operators
$\hat{q}$ and $\hat{p}$.

\begin{theorem}
If $f(q,p) = q^n \phi(p)$ where $n \in \mathbb{N}$ and $\phi$ is smooth on
$\overline{\mathbb{R}}$, then the corresponding operator $\hat{f}$ can be
defined by an analogical formula to \eqref{eq:32} and takes the form
\begin{equation}
\hat{f} = f_\lambda(\hat{q},\hat{p})
= \sum_{l=0}^n \binom{n}{l} \lambda^{l} (1 - \lambda)^{n - l}
    \hat{q}^l \phi(\hat{p}) \hat{q}^{n-l}.
\end{equation}
\end{theorem}

\begin{proof}
For $\varphi \in \mathcal{H}$ and $\psi \in \mathcal{S}$ we have that
$SW_\lambda(\varphi,\psi) = W_{1-\lambda}(\varphi,\psi)$. Hence
\begin{align}
\braket{f,W_\lambda(\varphi,\psi)} & = \int_{\mathbb{R}^2} q^n \phi(p)
    W_{1-\lambda}(\varphi,\psi)(q,p) \ud{q}\ud{\mu(p)}
    \nonumber \displaybreak[0] \\
& = \int_\mathbb{R} q^n \left( \int_\mathbb{R} \left( \int_\mathbb{R} \phi(p)
    \overline{\varphi(p \ominus \lambda \circ p')}
    \psi(p \oplus (1 - \lambda) \circ p') e^{\frac{i}{\hbar}(q,p')} \ud{\mu(p')}
    \right) \ud{\mu(p)} \right) \ud{q} \nonumber \displaybreak[0] \\
& = \int_\mathbb{R} q^n \left( \int_\mathbb{R} \left( \int_\mathbb{R} \phi(p)
    \overline{\varphi(p \ominus \lambda \circ p')}
    \psi(p \oplus (1 - \lambda) \circ p') e^{\frac{i}{\hbar}(q,p')} \ud{\mu(p)}
    \right) \ud{\mu(p')} \right) \ud{q} \nonumber \displaybreak[0] \\
& = \int_\mathbb{R} q^n \left( \int_\mathbb{R} \left( \int_\mathbb{R}
    \overline{\varphi(p)} \phi(p \oplus \lambda \circ p')
    \psi(p \oplus p') e^{\frac{i}{\hbar}(q,p')} \ud{\mu(p)} \right)
    \ud{\mu(p')} \right) \ud{q} \nonumber \displaybreak[0] \\
& = \int_\mathbb{R} q^n \left( \int_\mathbb{R} \left( \int_\mathbb{R}
    \overline{\varphi(p)} \phi(p \oplus \lambda \circ p')
    \psi(p \oplus p') e^{\frac{i}{\hbar}(q,p')} \ud{\mu(p')} \right)
    \ud{\mu(p)} \right) \ud{q} \nonumber \displaybreak[0] \\
& = \int_\mathbb{R} q^n \left( \int_\mathbb{R} \left( \int_\mathbb{R}
    \overline{\varphi(p)} \phi((1 - \lambda) \circ p \oplus \lambda \circ p')
    \psi(p') e^{\frac{i}{\hbar}(q,p')} \ud{\mu(p')} \right)
    e^{-\frac{i}{\hbar}(q,p)} \ud{\mu(p)} \right) \ud{q}
    \nonumber \displaybreak[0] \\
& = \int_\mathbb{R} \left( \int_\mathbb{R} \overline{\varphi(p)} \left(
    \int_\mathbb{R} (i\hbar D_{p'})^n
    \Bigl(\phi((1 - \lambda) \circ p \oplus \lambda \circ p') \psi(p')\Bigr)
    e^{\frac{i}{\hbar}(q,p')} \ud{\mu(p')} \right) e^{-\frac{i}{\hbar}(q,p)}
    \ud{\mu(p)} \right) \ud{q} \nonumber \displaybreak[0] \\
& = \int_\mathbb{R} \overline{\varphi(p)} \left( \int_\mathbb{R} \left(
    \int_\mathbb{R} (i\hbar D_{p'})^n
    \Bigl(\phi((1 - \lambda) \circ p \oplus \lambda \circ p') \psi(p')\Bigr)
    e^{\frac{i}{\hbar}(q,p')} \ud{\mu(p')} \right) e^{-\frac{i}{\hbar}(q,p)}
    \ud{q} \right) \ud{\mu(p)} \nonumber \displaybreak[0] \\
& = \int_\mathbb{R} \overline{\varphi(p)} (i\hbar D_{p'})^n
    \Bigl(\phi((1 - \lambda) \circ p \oplus \lambda \circ p')
    \psi(p')\Bigr) \bigg|_{p' = p} \ud{\mu(p)} \nonumber \displaybreak[0] \\
& = \int_\mathbb{R} \overline{\varphi(p)} (i\hbar D_{p'})^n
    \Bigl(\phi(p \oplus \lambda \circ p')
    \psi(p \oplus p')\Bigr) \bigg|_{p' = 0} \ud{\mu(p)}.
\end{align}
Thus, the bilinear form $(\varphi,\psi) \mapsto
\braket{f,W_\lambda(\varphi,\psi)}$ is continuous with respect to the first
variable and defines an operator $\hat{f}$ through the formula
\begin{equation}
\hat{f}\psi(p) = (i\hbar D_{p'})^n \Bigl(\phi(p \oplus \lambda \circ p')
    \psi(p \oplus p')\Bigr) \bigg|_{p' = 0}.
\end{equation}
By \eqref{eq:10} we get
\begin{align}
\hat{f}\psi(p) & = \sum_{k=0}^n \binom{n}{k} \Bigl((i\hbar D_{p'})^k
    \phi(p \oplus \lambda \circ p')\Bigr) \Bigl((i\hbar D_{p'})^{n-k}
    \psi(p \oplus p')\Bigr) \bigg|_{p' = 0} \nonumber \displaybreak[0] \\
& = \sum_{k=0}^n \binom{n}{k} \Bigl(\lambda^k (i\hbar D_p)^k
    \phi(p \oplus \lambda \circ p')\Bigr) \Bigl((i\hbar D_p)^{n-k}
    \psi(p \oplus p')\Bigr) \bigg|_{p' = 0} \nonumber \displaybreak[0] \\
& = \sum_{k=0}^n \binom{n}{k} \Bigl(\lambda^k (i\hbar D_p)^k \phi(p)\Bigr)
    \Bigl((i\hbar D_p)^{n-k} \psi(p)\Bigr) \nonumber \displaybreak[0] \\
& = \sum_{k=0}^n \binom{n}{k} \bigl(\lambda + (1 - \lambda)\bigr)^{n-k}
    \Bigl(\lambda^k (i\hbar D_p)^k \phi(p)\Bigr)
    \Bigl((i\hbar D_p)^{n-k} \psi(p)\Bigr) \nonumber \displaybreak[0] \\
& = \sum_{k=0}^n \binom{n}{k} \sum_{l=k}^n \binom{n-k}{l-k} \lambda^{l-k}
    (1 - \lambda)^{n - k - (l - k)} \Bigl(\lambda^k (i\hbar D_p)^k \phi(p)\Bigr)
    \Bigl((i\hbar D_p)^{n-k} \psi(p)\Bigr) \nonumber \displaybreak[0] \\
& = \sum_{l=0}^n \sum_{k=0}^l \binom{n}{l} \binom{l}{k} \lambda^{l}
    (1 - \lambda)^{n - l} \Bigl((i\hbar D_p)^k \phi(p)\Bigr)
    \Bigl((i\hbar D_p)^{l-k} (i\hbar D_p)^{n-l} \psi(p)\Bigr)
    \nonumber \displaybreak[0] \\
& = \sum_{l=0}^n \binom{n}{l} \lambda^{l} (1 - \lambda)^{n - l} (i\hbar D_p)^l
    \Bigl(\phi(p) (i\hbar D_p)^{n-l} \psi(p)\Bigr) \nonumber \displaybreak[0] \\
& = \sum_{l=0}^n \binom{n}{l} \lambda^{l} (1 - \lambda)^{n - l}
    \hat{q}^l \phi(\hat{p}) \hat{q}^{n-l} \psi(p).
\end{align}
\end{proof}

For $f \in \mathcal{F}(\mathbb{R}^2)$ the corresponding operator $\hat{f}$ can
be written in a form
\begin{equation}
\hat{f} = f_\lambda(\hat{q},\hat{p})
= \frac{1}{2\pi\hbar} \int_{\mathbb{R}^2} \mathcal{F}_\beta f(q',p')
    e^{\frac{i}{\hbar}\lambda(\hat{q},p')}
    e^{-\frac{i}{\hbar}(q',\hat{p})}
    e^{\frac{i}{\hbar}(1 - \lambda)(\hat{q},p')}
    \ud{q'}\ud{\mu(p')},
\end{equation}
where $e^{\frac{i}{\hbar}\lambda(\hat{q},p')}$ and
$e^{-\frac{i}{\hbar}(q',\hat{p})}$ are unitary operators which action on a
function $\psi \in \mathcal{H}$ is given by the formulas
\begin{equation}
e^{\frac{i}{\hbar}\lambda(\hat{q},p')}\psi(p) =
    \psi(p \ominus \lambda \circ p'), \quad
e^{-\frac{i}{\hbar}(q',\hat{p})}\psi(p) = e^{-\frac{i}{\hbar}(q',p)}\psi(p).
\end{equation}

\section{Position eigenvectors}
\label{sec:9}
Solutions to the eigenvalue equations
\begin{equation}
q \star \rho_\xi = \xi\rho_\xi, \quad \rho_\xi \star q = \xi\rho_\xi, \quad
\xi \in \mathbb{R}
\end{equation}
define position eigenvectors. They are in the form
\begin{equation}
\rho_\xi(q,p) = \sinc\left(\frac{q - \xi}{2\hbar\sqrt\beta}\right),
\end{equation}
where $\sinc x = \frac{\sin \pi x}{\pi x}$. Indeed, since
\begin{equation}
\tilde{\rho}_\xi(p',p) = 2\hbar\sqrt\beta e^{-\frac{i}{\hbar}(\xi,p')}
\end{equation}
we get
\begin{align}
(q \star \rho_\xi)^\sim(p',p) & = \frac{1}{2\pi\hbar} \int_\mathbb{R} q' \left(
    \int_\mathbb{R}
    \tilde{\rho}_\xi\bigl(p' \ominus p'',p \ominus (1 - \lambda) \circ p''\bigr)
    e^{-\frac{i}{\hbar}(q',p'')} \ud{\mu(p'')} \right) \ud{q'} \nonumber \\
& = \frac{\sqrt\beta}{\pi} \int_\mathbb{R} q' \left( \int_\mathbb{R}
    e^{-\frac{i}{\hbar}(\xi,p')} e^{-\frac{i}{\hbar}(q' - \xi,p'')}
    \ud{\mu(p'')} \right) \ud{q'} \nonumber \\
& = e^{-\frac{i}{\hbar}(\xi,p')} \int_\mathbb{R} q'
    \sinc\left(\frac{q' - \xi}{2\hbar\sqrt\beta}\right) \ud{q'} \nonumber \\
& = e^{-\frac{i}{\hbar}(\xi,p')} \lim_{R \to \infty} \int_{-R}^R (q' + \xi)
    \sinc\left(\frac{q'}{2\hbar\sqrt\beta}\right) \ud{q'} \nonumber \\
& = 2\hbar\sqrt\beta\xi e^{-\frac{i}{\hbar}(\xi,p')}
= \xi \tilde{\rho}_\xi(p',p)
\end{align}
and similarly
\begin{equation}
(\rho_\xi \star q)^\sim(p',p) = \xi \tilde{\rho}_\xi(p',p).
\end{equation}

The elements of the Hilbert space $\mathcal{H}$ corresponding to $\rho_\xi$ are
in the form
\begin{equation}
\psi_\xi(p) = \sqrt{\frac{\sqrt\beta}{\pi}} e^{-\frac{i}{\hbar}(\xi,p)}.
\end{equation}
A direct calculation shows that indeed $\rho_\xi = W_\lambda(\psi_\xi,\psi_\xi)$
and that $\psi_\xi$ are eigenvectors of the operator $\hat{q}$.

The position eigenvectors $\rho_\xi$ are not physical states as their
uncertainty in position is smaller than the minimal uncertainty $\Delta q_0$
(cf. \cite{Kempf:1995} where this is discussed in full detail).

\section{Maximal localization states}
\label{sec:10}
In \cite{Kempf:1995} authors calculated states of maximal localization around a
position $\xi$, i.e. states with the following properties
\begin{equation}
\braket{\hat{q}} = \xi, \quad \braket{\hat{p}} = 0, \quad \Delta q = \Delta q_0,
\end{equation}
where $\Delta q_0 = \hbar\sqrt\beta$ is the smallest uncertainty in position
which can be reached. In momentum representation the received states were given
by the formula
\begin{equation}
\psi_\xi^\text{ML}(p) = \sqrt{\frac{2\sqrt\beta}{\pi}} (1 + \beta p^2)^{-1/2}
    e^{-\frac{i}{\hbar}(\xi,p)}.
\end{equation}

To demonstrate the developed formalism we can calculate the form of the maximal
localization states on phase space. By virtue of \eqref{eq:34} we receive the
following formula for the quasi-probability distribution function representing
a state of maximal localization around a point $(\xi,0)$ in phase space:
\begin{align}
\rho_\xi^\text{ML}(q,p) & =
    W_\lambda(\psi_\xi^\text{ML},\psi_\xi^\text{ML})(q,p) \nonumber \\
& = \frac{1}{2} \frac{1 - \beta p^2}{1 + \beta p^2} \left(
    \sinc\left(\frac{1}{2} - \lambda - \frac{q - \xi}{2\hbar\sqrt\beta}
    \right)
    + \sinc\left(\frac{1}{2} - \lambda + \frac{q - \xi}{2\hbar\sqrt\beta}
    \right)\right) \nonumber \\
& \quad {} + \frac{1}{2}\left(
    \sinc\left(\frac{1}{2} - \frac{q - \xi}{2\hbar\sqrt\beta}\right)
    + \sinc\left(\frac{1}{2} + \frac{q - \xi}{2\hbar\sqrt\beta}\right)\right)
    \nonumber \\
& \quad {} + \frac{i\sqrt\beta p}{1 + \beta p^2} \left(
    \sinc\left(\frac{1}{2} - \lambda - \frac{q - \xi}{2\hbar\sqrt\beta}
    \right)
    - \sinc\left(\frac{1}{2} - \lambda + \frac{q - \xi}{2\hbar\sqrt\beta}
    \right)\right).
\label{eq:35}
\end{align}
Indeed, we have that
\begin{align}
& \overline{\psi_\xi^\text{ML}\bigl(p \ominus (1 - \lambda) \circ p'\bigr)}
    \psi_\xi^\text{ML}\bigl(p \oplus \lambda \circ p'\bigr)
    e^{\frac{i}{\hbar}(q,p')} = \nonumber \\
& \qquad \frac{2\sqrt\beta}{\pi}
    \bigl(1 + \beta (p \ominus (1 - \lambda) \circ p')^2\bigr)^{-1/2}
    \bigl(1 + \beta (p \oplus \lambda \circ p')^2\bigr)^{-1/2}
    e^{\frac{i}{\hbar}(q - \xi,p')}.
\end{align}
Since
\begin{equation}
1 + \beta (p \oplus \lambda \circ p')^2 = (1 + \beta p^2)
    \frac{1 + \beta (\lambda \circ p')^2}{(1 - \beta p (\lambda \circ p'))^2}
= (1 + \beta p^2) \frac{1 + \tan^2(\lambda\arctan(\sqrt\beta p'))}
    {1 - \sqrt\beta p \tan(\lambda\arctan(\sqrt\beta p'))}
\end{equation}
we get for $p' = \frac{1}{\sqrt\beta}\tan\bar{p}$, after using various
trigonometric identities, that
\begin{align}
& \overline{\psi_\xi^\text{ML}\bigl(p \ominus (1 - \lambda) \circ p'\bigr)}
    \psi_\xi^\text{ML}\bigl(p \oplus \lambda \circ p'\bigr)
    e^{\frac{i}{\hbar}(q,p')} = \nonumber \\
& \qquad \frac{2\sqrt\beta}{\pi}
    \left(\frac{1}{2}\frac{1 - \beta p^2}{1 + \beta p^2}
    \cos\bigl((1 - 2\lambda)\bar{p}\bigr) + \frac{1}{2} \cos(\bar{p})
    + \frac{\sqrt\beta p}{1 + \beta p^2} \sin\bigl((1 - 2\lambda)\bar{p}\bigr)
    \right) e^{i\frac{q - \xi}{\hbar\sqrt\beta}\bar{p}}.
\end{align}
Hence, the function $\rho_\xi^\text{ML}(q,p)$ is expressed by the following
integrals
\begin{align}
\rho_\xi^\text{ML}(q,p) & = \frac{1}{\pi}\Biggl(
    \frac{1 - \beta p^2}{1 + \beta p^2}
    \int_{-\frac{\pi}{2}}^{\frac{\pi}{2}} \cos\bigl((1 - 2\lambda)\bar{p}\bigr)
    e^{i\frac{q - \xi}{\hbar\sqrt\beta}\bar{p}} \ud{\bar{p}}
    + \int_{-\frac{\pi}{2}}^{\frac{\pi}{2}} \cos(\bar{p})
    e^{i\frac{q - \xi}{\hbar\sqrt\beta}\bar{p}} \ud{\bar{p}} \nonumber \\
& \quad {} + \frac{2\sqrt\beta p}{1 + \beta p^2}
    \int_{-\frac{\pi}{2}}^{\frac{\pi}{2}} \sin\bigl((1 - 2\lambda)\bar{p}\bigr)
    e^{i\frac{q - \xi}{\hbar\sqrt\beta}\bar{p}} \ud{\bar{p}} \Biggr),
\end{align}
which can be easily calculated to give \eqref{eq:35}.

Note, that the state $\rho_\xi^\text{ML}(q,p)$ is a shift in position by $\xi$
of the state $\rho_0^\text{ML}(q,p)$ localized around a point $(0,0)$ in phase
space. For particular values of $\lambda$ the states $\rho_\xi^\text{ML}(q,p)$
take a simpler form. For instance, when $\lambda = \tfrac{1}{2}$ we have
\begin{equation}
\rho_\xi^\text{ML}(q,p) = \frac{1 - \beta p^2}{1 + \beta p^2}
    \sinc\left(\frac{q - \xi}{2\hbar\sqrt\beta}\right) + \frac{1}{2}\left(
    \sinc\left(\frac{1}{2} - \frac{q - \xi}{2\hbar\sqrt\beta}\right)
    + \sinc\left(\frac{1}{2} + \frac{q - \xi}{2\hbar\sqrt\beta}\right)\right).
\end{equation}
In this case the quasi-probability distribution functions
$\rho_\xi^\text{ML}(q,p)$ are real-valued, which is the common property of every
state. The plot of the state localized around a point $(0,0)$ is presented in
Fig.~\ref{fig:1}. For $\lambda = 0$ we have
\begin{align}
\rho_\xi^\text{ML}(q,p) & = \frac{1}{1 + \beta p^2}\left(
    \sinc\left(\frac{1}{2} - \frac{q - \xi}{2\hbar\sqrt\beta}\right)
    + \sinc\left(\frac{1}{2} + \frac{q - \xi}{2\hbar\sqrt\beta}\right)\right)
    \nonumber \\
& \quad {} + \frac{i\sqrt\beta p}{1 + \beta p^2} \left(
    \sinc\left(\frac{1}{2} - \frac{q - \xi}{2\hbar\sqrt\beta}\right)
    - \sinc\left(\frac{1}{2} + \frac{q - \xi}{2\hbar\sqrt\beta}\right)\right).
\end{align}
The plots of real and imaginary parts of the state localized around a point
$(0,0)$ are presented in Fig.~\ref{fig:2}.

\begin{figure}
\centering
\includegraphics{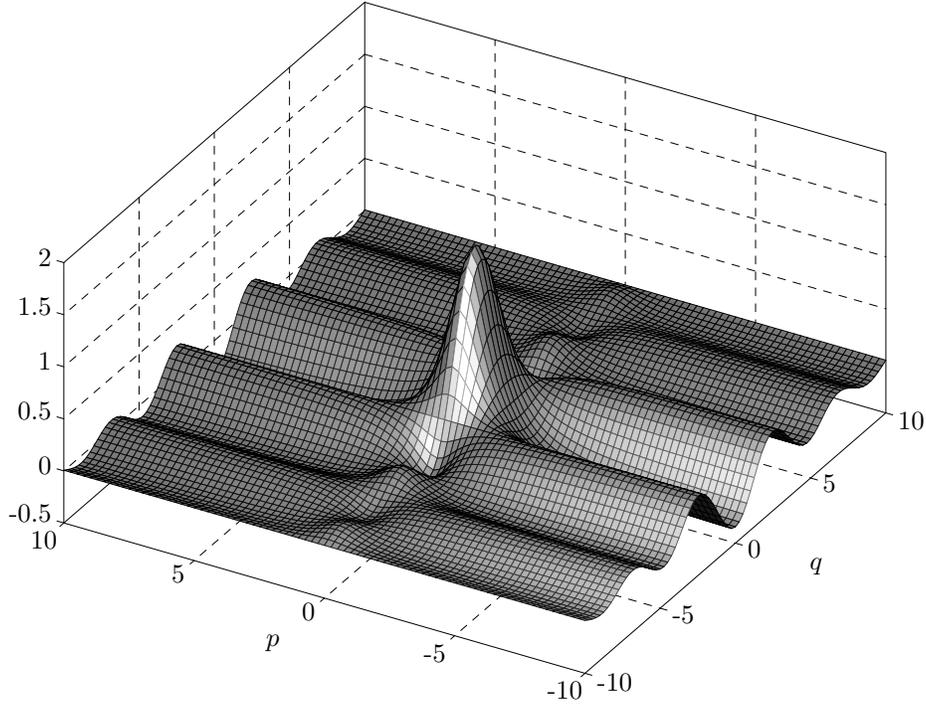}
\caption{Plot of maximally localized state $\rho_0^\text{ML}(q,p)$ for
$\lambda = \tfrac{1}{2}$. Units $\hbar = \beta = 1$ are used.}
\label{fig:1}
\end{figure}

\begin{figure}
\centering
\includegraphics{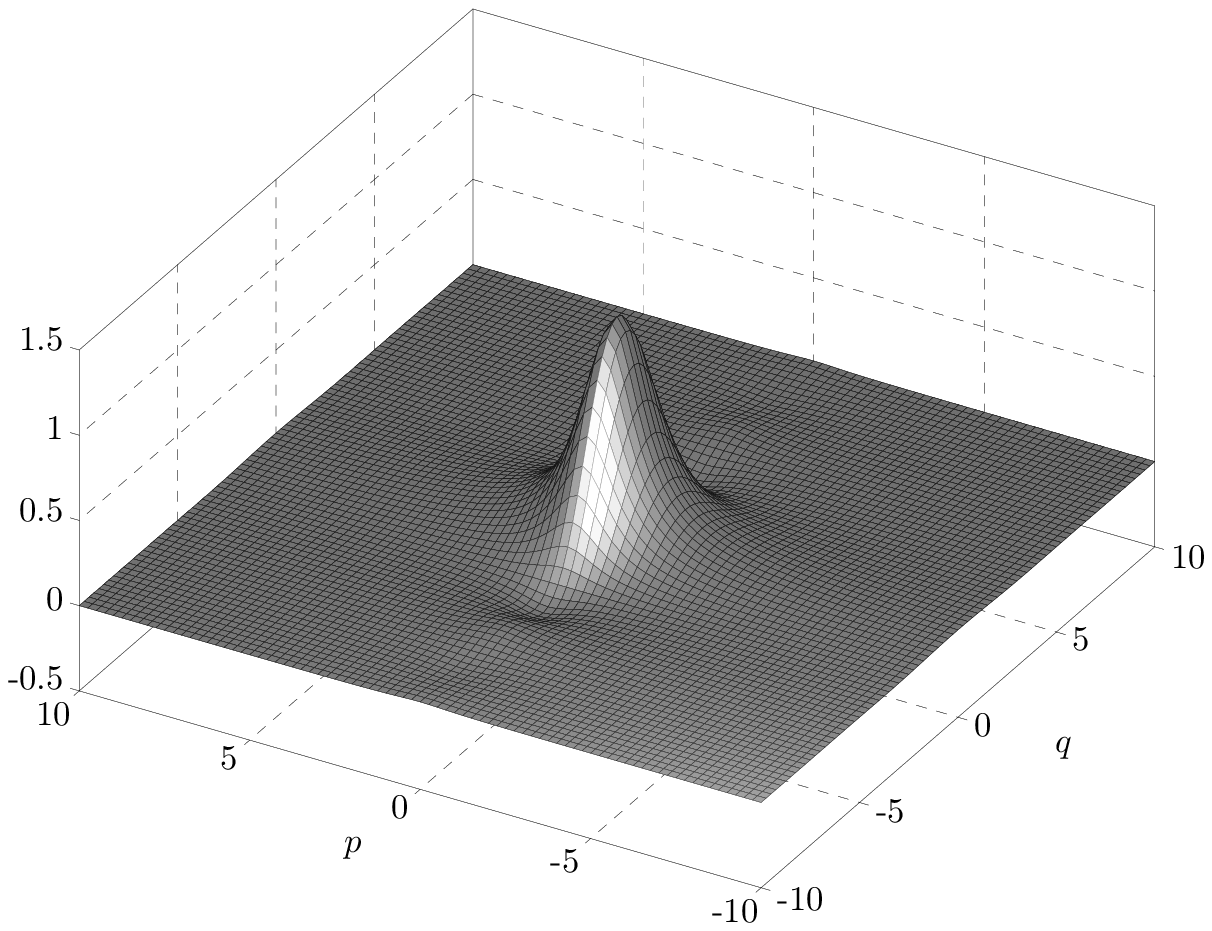} \\
\bigskip
\includegraphics{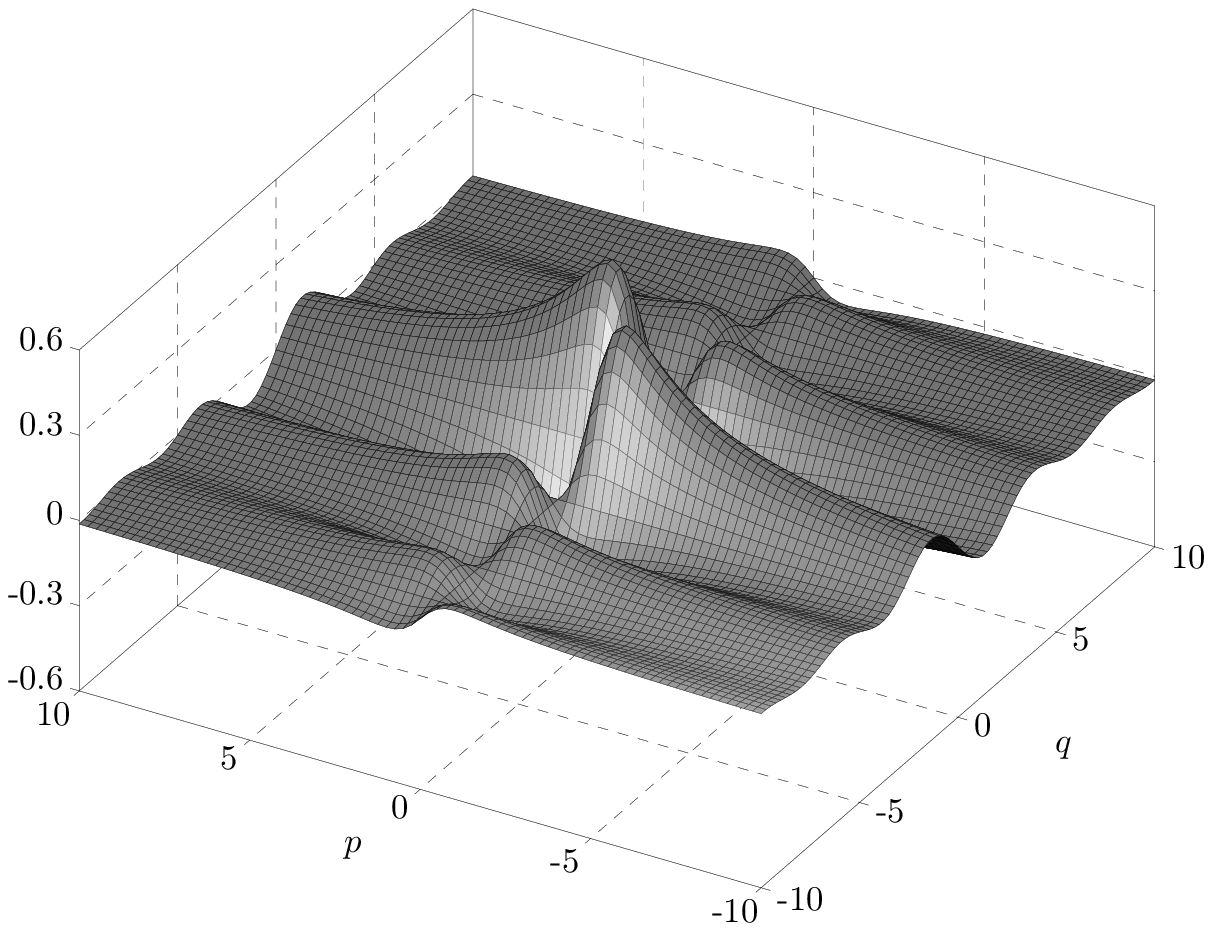}
\caption{Plot of maximally localized state $\rho_0^\text{ML}(q,p)$ for
$\lambda = 0$. Real part in upper plot, imaginary part in lower plot. Units
$\hbar = \beta = 1$ are used.}
\label{fig:2}
\end{figure}

\section{Conclusions and final remarks}
\label{sec:11}
In the paper we developed a formalism of non-formal deformation quantization
exhibiting a minimal length scale. The theory was presented for the case of
two dimensional phase space $\mathbb{R}^2$. The generalization to $2n$
dimensions leading to the following commutation relations for the operators of
position and momentum
\begin{equation}
\begin{split}
[\hat{q}_j,\hat{p}_k] & = i\hbar\delta_{jk}
    \left(\hat{1} + \beta \hat{p}_j \hat{p}_k\right), \\
[\hat{q}_j,\hat{q}_k] & = [\hat{p}_j,\hat{p}_k] = 0
\end{split}
\end{equation}
is straightforward. However, if we would like to consider different commutation
relations, like the one found in \cite{Kempf:1995}
\begin{equation}
\begin{split}
[\hat{q}_j,\hat{p}_k] & = i\hbar\delta_{jk}
    \left(\hat{1} + \beta \hat{p}^2\right), \\
[\hat{q}_j,\hat{q}_k] & = 2i\hbar\beta
    \bigl(\hat{p}_j \hat{q}_k - \hat{p}_k \hat{q}_j\bigr), \\
[\hat{p}_j,\hat{p}_k] & = 0,
\end{split}
\end{equation}
where $\hat{p}^2 = \sum_{i=1}^n \hat{p}_i^2$, then the presented theory does not
seem to generalize in a simple way.

An interesting topic of further investigation would be a development of
deformation quantization theory with minimal length based on more general
uncertainty relations exhibiting also a minimal uncertainty in momentum, like
the one considered in \cite{Kempf:1992}
\begin{equation}
\Delta q \Delta p \geq \frac{\hbar}{2}\left(1 + \alpha (\Delta q)^2
    + \alpha \braket{\hat{q}}^2 + \beta (\Delta p)^2
    + \beta \braket{\hat{p}}^2\right).
\end{equation}
However, in this case there is neither position nor momentum operator
representation, and one have to use Bergmann-Fock representation or only a
deformation quantization approach.

The \nbr{\star}product \eqref{eq:4} is not the only product exhibiting a minimal
uncertainty in position which can be introduced. Another example of such a
star-product is
\begin{equation}
f \star g = f\exp\left(i\hbar(1 - \lambda)\overleftarrow{\mathcal{D}_q}
    \overrightarrow{\mathcal{D}_p}
    - i\hbar\lambda\overleftarrow{\mathcal{D}_p}
    \overrightarrow{\mathcal{D}_q}\right)g,
\end{equation}
where
\begin{equation}
\begin{split}
\mathcal{D}_q & = (1 + \beta p^2)^{-1/2} \partial_q, \\
\mathcal{D}_p & = -\beta qp(1 + \beta p^2)^{1/2} \partial_q
    + (1 + \beta p^2)^{3/2} \partial_p.
\end{split}
\end{equation}
It can be easily verified that it satisfies the commutation relation
\eqref{eq:43}. This star-product can be transformed to the product of the form
\eqref{eq:3} by the following noncanonical transformation of coordinates
\begin{equation}
\begin{split}
\bar{q} & = q \sqrt{1 + \beta p^2}, \\
\bar{p} & = \frac{p}{\sqrt{1 + \beta p^2}}.
\end{split}
\end{equation}
Note, however, that contrary to the star-product \eqref{eq:4} for this
star-product the property \eqref{eq:42} will not hold. In fact, we can calculate
that
\begin{equation}
q \star q = q^2 + i\hbar(2\lambda - 1)\beta qp
    + \frac{1}{2}\hbar^2 \lambda(1 - \lambda)\beta^2 p^2, \quad
p \star p = p^2.
\end{equation}
It seems that the star-products from the family \eqref{eq:4} are the only
products with the property \eqref{eq:42}. Thus, with respect to this property,
the family \eqref{eq:4} of star-products is distinguished from other products
satisfying the commutation relation \eqref{eq:43}.


\begin{thebibliography}{34}
\providecommand{\natexlab}[1]{#1}
\providecommand{\url}[1]{\texttt{#1}}
\providecommand{\urlprefix}{URL }
\providecommand{\selectlanguage}[1]{\relax}
\providecommand{\eprint}[2][]{\url{#2}}

\bibitem[{Snyder(1947)}]{Snyder:1947}
H.~S. Snyder.
\newblock \emph{{Quantized space-time}}.
\newblock Phys. Rev. \textbf{71}(1), (1947) pp. 38--41

\bibitem[{Mead(1964)}]{Mead:1964}
C.~A. Mead.
\newblock \emph{{Possible connection between gravitation and fundamental
  length}}.
\newblock Phys. Rev. \textbf{135}(3B), (1964) pp. B849--B862

\bibitem[{Mead(1966)}]{Mead:1966}
C.~A. Mead.
\newblock \emph{{Observable consequences of fundamental-length hypotheses}}.
\newblock Phys. Rev. \textbf{143}(4), (1966) pp. 990--1005

\bibitem[{Adler and Santiago(1999)}]{Adler:1999}
R.~J. Adler and D.~I. Santiago.
\newblock \emph{{On gravity and the uncertainty principle}}.
\newblock Modern Phys. Lett. A \textbf{14}(20), (1999) pp. 1371--1381.
\newblock \eprint{arXiv:gr-qc/9904026}

\bibitem[{Amati et~al.(1987)Amati, Ciafaloni, and Veneziano}]{Amati:1987}
D.~Amati, M.~Ciafaloni, and G.~Veneziano.
\newblock \emph{{Superstring collisions at planckian energies}}.
\newblock Phys. Lett. B \textbf{197}(1--2), (1987) pp. 81--88

\bibitem[{Amati et~al.(1988)Amati, Ciafaloni, and Veneziano}]{Amati:1988}
D.~Amati, M.~Ciafaloni, and G.~Veneziano.
\newblock \emph{{Classical and quantum gravity effects from planckian energy
  superstring collisions}}.
\newblock Int. J. Mod. Phys. A \textbf{3}(7), (1988) pp. 1615--1661

\bibitem[{Gross and Mende(1988)}]{Gross:1988}
D.~J. Gross and P.~F. Mende.
\newblock \emph{{String theory beyond the Planck scale}}.
\newblock Nucl. Phys. B \textbf{303}(3), (1988) pp. 407--454

\bibitem[{Amati et~al.(1990)Amati, Ciafaloni, and Veneziano}]{Amati:1990}
D.~Amati, M.~Ciafaloni, and G.~Veneziano.
\newblock \emph{{Higher order gravitational deflection and soft bremsstrahlung
  in planckian energy superstring collisions}}.
\newblock Nucl. Phys. B \textbf{347}(3), (1990) pp. 550--580

\bibitem[{Konishi et~al.(1990)Konishi, Paffuti, and Provero}]{Konishi:1990}
K.~Konishi, G.~Paffuti, and P.~Provero.
\newblock \emph{{Minimum physical length and the generalized uncertainty
  principle in string theory}}.
\newblock Phys. Lett. B \textbf{234}(3), (1990) pp. 276--284

\bibitem[{Maggiore(1993)}]{Maggiore:1993}
M.~Maggiore.
\newblock \emph{{The algebraic structure of the generalized uncertainty
  principle}}.
\newblock Phys. Lett. B \textbf{319}(1--3), (1993) pp. 83--86.
\newblock \eprint{arXiv:hep-th/9309034}

\bibitem[{Maggiore(1994)}]{Maggiore:1994}
M.~Maggiore.
\newblock \emph{{Quantum groups, gravity, and the generalized uncertainty
  principle}}.
\newblock Phys. Rev. D \textbf{49}(10), (1994) p. 5182.
\newblock \eprint{arXiv:hep-th/9305163}

\bibitem[{Scardigli(1999)}]{Scardigli:1999}
F.~Scardigli.
\newblock \emph{{Generalized uncertainty principle in quantum gravity from
  micro-black hole gedanken experiment}}.
\newblock Phys. Lett. B \textbf{452}(1--2), (1999) pp. 39--44.
\newblock \eprint{arXiv:hep-th/9904025}

\bibitem[{Scardigli and Casadio(2003)}]{Scardigli:2003}
F.~Scardigli and R.~Casadio.
\newblock \emph{{Generalized uncertainty principle, extra dimensions and
  holography}}.
\newblock Class. Quant. Grav. \textbf{20}(18), (2003) pp. 3915--3926.
\newblock \eprint{arXiv:hep-th/0307174}

\bibitem[{Hossenfelder(2013)}]{Hossenfelder:2013}
S.~Hossenfelder.
\newblock \emph{{Minimal length scale scenarios for quantum gravity}}.
\newblock Living Rev. Relativ. \textbf{16}(1), (2013) p.~2.
\newblock \eprint{arXiv:1203.6191 [gr-qc]}

\bibitem[{Kempf et~al.(1995)Kempf, Mangano, and Mann}]{Kempf:1995}
A.~Kempf, G.~Mangano, and R.~B. Mann.
\newblock \emph{{Hilbert space representation of the minimal length uncertainty
  relation}}.
\newblock Phys. Rev. D \textbf{52}(2), (1995) pp. 1108--1118.
\newblock \eprint{arXiv:hep-th/9412167}

\bibitem[{Kempf(1992)}]{Kempf:1992}
A.~Kempf.
\newblock \emph{{Quantum group-symmetric fock spaces with bargmann-fock
  representation}}.
\newblock Lett. Math. Phys. \textbf{26}(1), (1992) pp. 1--12

\bibitem[{Kempf(1994)}]{Kempf:1994}
A.~Kempf.
\newblock \emph{{Uncertainty relation in quantum mechanics with quantum group
  symmetry}}.
\newblock J. Math. Phys. \textbf{35}(9), (1994) pp. 4483--4496.
\newblock \eprint{arXiv:hep-th/9311147}

\bibitem[{Quesne et~al.(2003)Quesne, Penson, and Tkachuk}]{Quesne.Penson:2003}
C.~Quesne, K.~A. Penson, and V.~M. Tkachuk.
\newblock \emph{{Maths-type $q$-deformed coherent states for $q > 1$}}.
\newblock Phys. Lett. A \textbf{313}(1--2), (2003) pp. 29--36.
\newblock \eprint{arXiv:quant-ph/0303120}

\bibitem[{Quesne and Tkachuk(2003)}]{Quesne:2003}
C.~Quesne and V.~M. Tkachuk.
\newblock \emph{{Harmonic oscillator with nonzero minimal uncertainties in both
  position and momentum in a SUSYQM framework}}.
\newblock J. Phys. A \textbf{36}(41), (2003) pp. 10373--10389.
\newblock \eprint{arXiv:math-ph/0306047}

\bibitem[{Scardigli and Casadio(2015)}]{Scardigli:2015}
F.~Scardigli and R.~Casadio.
\newblock \emph{{Gravitational tests of the generalized uncertainty
  principle}}.
\newblock Eur. Phys. J. C \textbf{75}(9), (2015) p. 425.
\newblock \eprint{arXiv:1407.0113 [hep-th]}

\bibitem[{Scardigli et~al.(2017)Scardigli, Lambiase, and
  Vagenas}]{Scardigli:2017}
F.~Scardigli, G.~Lambiase, and E.~C. Vagenas.
\newblock \emph{{GUP parameter from quantum corrections to the Newtonian
  potential}}.
\newblock Phys. Lett. B \textbf{767}, (2017) pp. 242--246.
\newblock \eprint{arXiv:1611.01469 [hep-th]}

\bibitem[{Amelino-Camelia and Arzano(2002)}]{Amelino-Camelia:2002}
G.~Amelino-Camelia and M.~Arzano.
\newblock \emph{{Coproduct and star product in field theories on Lie-algebra
  noncommutative space-times}}.
\newblock Phys. Rev. D \textbf{65}(8), (2002) p. 084044.
\newblock \eprint{arXiv:hep-th/0105120}

\bibitem[{Freidel et~al.(2007)Freidel, Kowalski-Glikman, and
  Nowak}]{Freidel:2007}
L.~Freidel, J.~Kowalski-Glikman, and S.~Nowak.
\newblock \emph{{From noncommutative $\kappa$-Minkowski to Minkowski
  space-time}}.
\newblock Phys. Lett. B \textbf{648}(1), (2007) pp. 70--75.
\newblock \eprint{arXiv:hep-th/0612170}

\bibitem[{Meljanac and Kre{\v{s}}i{\'c}-Juri{\'c}(2008)}]{Meljanac:2008}
S.~Meljanac and S.~Kre{\v{s}}i{\'c}-Juri{\'c}.
\newblock \emph{{Generalized kappa-deformed spaces, star-products and their
  realizations}}.
\newblock J. Phys. A \textbf{41}(23), (2008) p. 235203.
\newblock \eprint{arXiv:0804.3072 [hep-th]}

\bibitem[{Bayen et~al.(1978{\natexlab{a}})Bayen, Flato, Fr{\o}nsdal,
  Lichnerowicz, and Sternheimer}]{Bayen:1978a}
F.~Bayen, M.~Flato, C.~Fr{\o}nsdal, A.~Lichnerowicz, and D.~Sternheimer.
\newblock \emph{{Deformation theory and quantization. I. Deformations of
  symplectic structures}}.
\newblock Ann. Phys. \textbf{111}(1), (1978{\natexlab{a}}) pp. 61--110

\bibitem[{Bayen et~al.(1978{\natexlab{b}})Bayen, Flato, Fr{\o}nsdal,
  Lichnerowicz, and Sternheimer}]{Bayen:1978b}
F.~Bayen, M.~Flato, C.~Fr{\o}nsdal, A.~Lichnerowicz, and D.~Sternheimer.
\newblock \emph{{Deformation theory and quantization. II. Physical
  applications}}.
\newblock Ann. Phys. \textbf{111}(1), (1978{\natexlab{b}}) pp. 111--151

\bibitem[{B{\l}aszak and Doma{\'n}ski(2012)}]{Blaszak:2012}
M.~B{\l}aszak and Z.~Doma{\'n}ski.
\newblock \emph{{Phase space quantum mechanics}}.
\newblock Ann. Phys. \textbf{327}(2), (2012) pp. 167--211.
\newblock \eprint{arXiv:1009.0150 [math-ph]}

\bibitem[{de~Gosson(2006)}]{Gosson:2006}
M.~de~Gosson.
\newblock \emph{{Symplectic Geometry and Quantum Mechanics}}, volume 166 of
  \emph{{Operator Theory: Advances and Applications}}.
\newblock Birkh{\"a}user, Basel (2006)

\bibitem[{Bastos et~al.(2008)Bastos, Bertolami, Dias, and Prata}]{Bastos:2008}
C.~Bastos, O.~Bertolami, N.~C. Dias, and J.~N. Prata.
\newblock \emph{{Weyl-Wigner formulation of noncommutative quantum mechanics}}.
\newblock J. Math. Phys. \textbf{49}(7), (2008) p. 072101.
\newblock \eprint{arXiv:hep-th/0611257}

\bibitem[{Bastos et~al.(2010)Bastos, Dias, and Prata}]{Bastos:2010}
C.~Bastos, N.~C. Dias, and J.~N. Prata.
\newblock \emph{{Wigner measures in noncommutative quantum mechanics}}.
\newblock Commun. Math. Phys. \textbf{299}(3), (2010) pp. 709--740.
\newblock \eprint{arXiv:0907.4438 [math-ph]}

\bibitem[{Dias et~al.(2010)Dias, de~Gosson, Luef, and Prata}]{Dias:2010}
N.~C. Dias, M.~de~Gosson, F.~Luef, and J.~N. Prata.
\newblock \emph{{A deformation quantization theory for noncommutative quantum
  mechanics}}.
\newblock J. Math. Phys. \textbf{51}(7), (2010) pp. 072101--072112

\bibitem[{Tkachuk(2012)}]{Tkachuk:2012}
V.~M. Tkachuk.
\newblock \emph{{Deformed Heisenberg algebra with minimal length and the
  equivalence principle}}.
\newblock Phys. Rev. A \textbf{86}(6), (2012) p. 062112.
\newblock \eprint{arXiv:1301.1891 [gr-qc]}

\bibitem[{Czachor(2016)}]{Czachor:2016}
M.~Czachor.
\newblock \emph{{Relativity of arithmetic as a fundamental symmetry of
  physics}}.
\newblock Quantum Stud.: Math. Found. \textbf{3}(2), (2016) pp. 123--133.
\newblock \eprint{arXiv:1412.8583 [math-ph]}

\bibitem[{Sugiura(1990)}]{Sugiura:1990}
M.~Sugiura.
\newblock \emph{{Unitary Representations and Harmonic Analysis}}, volume~44 of
  \emph{North-Holland Mathematical Library}.
\newblock North-Holland Publishing Co., Amsterdam, second edition (1990)
\end{thebibliography}

\end{document}